\documentclass[a4paper,onecolumn,11pt,accepted=2023-12-18]{quantumarticle}
\pdfoutput=1
\usepackage{amsmath,amssymb,amsfonts}
\usepackage{dcolumn}
\usepackage{algorithmic}
\usepackage{graphicx}
\usepackage{textcomp}
\usepackage{xcolor}
\usepackage{soul}
\usepackage{url}
\usepackage{bm}
\usepackage{float}
\usepackage{hyperref}
\usepackage{multirow}
\usepackage{ragged2e}
\usepackage{caption} 
\usepackage{hhline}
\usepackage{braket}
\usepackage{mathtools}
\usepackage{listings}
\usepackage{array}
\usepackage{amsthm}
\newcolumntype{P}[1]{>{\centering\arraybackslash}p{#1}}
\def\BibTeX{{\rm B\kern-.05em{\sc i\kern-.025em b}\kern-.08em
    T\kern-.1667em\lower.7ex\hbox{E}\kern-.125emX}}

\newcommand{\zv}{\boldsymbol{z}}

\newcounter{subeqns}
\newcommand{\subeqns}{(\roman{subeqns})}
\newcommand{\subnumber}{\refstepcounter{subeqns}\subeqns}
\newcommand{\lowerromannumeral}[1]{\romannumeral#1\relax}

\newcommand{\qaoaobjsingle}{\langle C(\gamma) \rangle}

\newtheorem{theorem}{Theorem}
\newcommand{\introthmnum}{3}
\newtheorem*{theoremintro}{Theorem \introthmnum{}} %

\newtheorem{corollary}{Corollary}[theorem]
\newtheorem{lemma}{Lemma}[theorem]

\newtheorem{remark}{Remark}

\makeatletter
\def\@opargbegintheorem#1#2#3{\trivlist \item[]{\bfseries #1\ #2\ (#3)} \itshape}
\makeatother

\newcommand{\Gate}[1]{\textsc{#1}}

\newcommand{\zgate}{\Gate{z}}
\newcommand{\ygate}{\Gate{y}}
\newcommand{\xgate}{\Gate{x}}

\newcommand{\idgate}{\Gate{i}}

\setstcolor{green}

\begin{document}

\title{Parameter Setting in Quantum Approximate Optimization of Weighted Problems}

\author{Shree Hari Sureshbabu\textsuperscript{$\ddagger$}}%
\affiliation{Global Technology Applied Research, JPMorgan Chase, New York, NY 10017}
\author{Dylan Herman\textsuperscript{$\ddagger$}}%
\affiliation{Global Technology Applied Research, JPMorgan Chase, New York, NY 10017}
\author{Ruslan Shaydulin\textsuperscript{$\ddagger$}}%
\affiliation{Global Technology Applied Research, JPMorgan Chase, New York, NY 10017}
\author{Joao Basso}
\affiliation{Department of Mathematics, University of California, Berkeley, CA 94720}
\author{Shouvanik Chakrabarti}
\affiliation{Global Technology Applied Research, JPMorgan Chase, New York, NY 10017}
\author{Yue Sun}
\affiliation{Global Technology Applied Research, JPMorgan Chase, New York, NY 10017}
\author{Marco Pistoia}
\affiliation{Global Technology Applied Research, JPMorgan Chase, New York, NY 10017}

\footnotetext{Equal contribution.}

\begin{abstract}
Quantum Approximate Optimization Algorithm (QAOA) is a leading candidate algorithm for solving combinatorial optimization problems on quantum computers. However, in many cases QAOA requires computationally intensive parameter optimization. The challenge of parameter optimization is particularly acute in the case of weighted problems, for which the eigenvalues of the phase operator are non-integer and the QAOA energy landscape is not periodic. In this work, we develop parameter setting heuristics for QAOA applied to a general class of weighted problems. First,  we derive optimal parameters for QAOA with depth $p=1$ applied to the weighted MaxCut problem under different assumptions on the weights. In particular, we rigorously prove the conventional wisdom that in the average case the first local optimum near zero gives globally-optimal QAOA parameters. Second, for $p\geq 1$ we prove that the QAOA energy landscape for weighted MaxCut approaches that for the unweighted case under a simple rescaling of parameters. Therefore, we can use parameters previously obtained for unweighted MaxCut for weighted problems. Finally, we prove that for $p=1$ the QAOA objective sharply concentrates around its expectation, which means that our parameter setting rules hold with high probability for a random weighted instance. We numerically validate this approach on general weighted graphs and show that on average the QAOA energy with the proposed fixed parameters is only $1.1$ percentage points away from that with optimized parameters. Third, we propose a general heuristic rescaling scheme inspired by the analytical results for weighted MaxCut and demonstrate its effectiveness using QAOA with the $\xgate\ygate$ Hamming-weight-preserving mixer applied to the portfolio optimization problem. Our heuristic improves the convergence of local optimizers, reducing the number of iterations by 7.4x on average.
\end{abstract}

\maketitle

\section{Introduction}
Quantum computers are widely believed to be able to provide computational speedups for various problems of relevance to science and industry~\cite{nielsen2010quantum,herman2022}. Combinatorial optimization is a domain that is very likely to benefit from quantum computing due to the ubiquity of hard optimization problems. Quantum Approximate Optimization Algorithm (QAOA)~\cite{Hogg2000,farhi2014quantum,hadfield2019quantum} is a leading candidate quantum heuristic algorithm for optimization. QAOA solves optimization problems by preparing a parameterized quantum state using a circuit consisting of layers of alternating operators, wherein each operator has a free parameter associated with it. The two operators are commonly referred to as the \emph{phase operator} and \emph{mixer operator}, respectively. QAOA has been shown to achieve better scaling than state-of-the-art classical solvers for finding exact solutions of $k$-SAT~\cite{2208.06909} and to achieve approximation ratios competitive with those of the best known classical algorithms for the unweighted MaxCut problem~\cite{lipics.tqc.2022.7,2111.12641}. When the mixer is different from the transverse field used in~\cite{Hogg2000,farhi2014quantum}, the algorithm is sometimes referred to as the Quantum Alternating Operator Ansatz~\cite{hadfield2019quantum}. Throughout the paper, we will use QAOA to refer to both without making a distinction.

One of the central challenges of applying QAOA to practically-relevant problems is the need to set the QAOA parameters. The parameter setting is particularly challenging for problems with objectives containing non-integer coefficients (weights) on the binary variables. %
The non-integer weights lead to the optimization landscape becoming non-periodic and in general hard to optimize~\cite{Shaydulin2023,boulebnane2023peptide,brandhofer2022benchmarking}. While parameter setting schemes leveraging analytically-optimal QAOA parameters exist, they are only available in a limited number of cases. For example, in the infinite-size limit, optimal QAOA parameters are known for unweighted
MaxCut on large-girth regular~\cite{lipics.tqc.2022.7}, Erd\H{o}s-R\'enyi~\cite{2110.10685} and other~\cite{2110.10685} random graph ensembles, as well as for the Sherrington-Kirkpatrick (SK) model~\cite{farhi2022quantum}. While a similar parameter setting scheme has been proposed and numerically validated for weighted MaxCut~\cite{Shaydulin2023}, no analytical results are known for weighted MaxCut or other problems with non-integer eigenvalues. 

In this work, we develop parameter setting heuristics for QAOA applied to a broad class of weighted problems. Our starting point is QAOA with the transverse-field mixer applied to the weighted MaxCut problem on large-girth regular graphs. We begin with $p=1$ and derive globally-optimal parameters for QAOA applied to graphs with edge weights drawn i.i.d. from the exponential distribution for any graph size and from an arbitrary distribution in the infinite-size limit. Our analysis rigorously proves the folklore notion that for problems with non-periodic QAOA energy landscapes, the first local optimum near zero contains globally-optimal parameters in the average case~\cite{Shaydulin2023,boulebnane2023peptide,brandhofer2022benchmarking}\footnote{We note the importance of the ``average case'' qualifier, as examples can be constructed where the optimal QAOA parameters are arbitrarily far from origin; see e.g. \cite[Figure 2]{Shaydulin2023}.}.
We then analyze the case of $p\geq 1$ and connect QAOA on  weighted MaxCut problems to QAOA on unweighted MaxCut by proving the following Theorem:
\begin{theoremintro}[Informal] Consider QAOA with depth $p$ and a regular graph $G$ with girth $>2p+1$ and i.i.d. random edge weights drawn from $\mathbf{w}$. Then the QAOA objective for weighted MaxCut on $G$ at parameters $\left(\bm\beta, \frac{\bm{\gamma}}{\sqrt{\mathbb{E}_{\mathbf{w}}[w^2]}}\right)$ is, up to a global scaling factor, equal to the QAOA objective for the corresponding MaxCut problem on the unweighted version of the same graph at parameters $\left(\bm\beta, \bm{\gamma}\right)$ on average in the infinite-size limit. Here the average is taken over the random choice of edge weights, and $\mathbb{E}_{\mathbf{w}}[w^2]$ is the second moment of the distribution from which the edge weights are drawn.
\end{theoremintro}
This result proves that the parameters that are optimal for unweighted MaxCut can be rescaled to be optimal for weighted MaxCut. As a consequence, it establishes a rule for setting parameters in QAOA for weighted MaxCut using the parameters for the unweighted case obtained previously in Ref.~\cite{lipics.tqc.2022.7}. We further prove that for $p=1$ the QAOA objective concentrates sharply around its expectation over randomly drawn weights. This implies that the derived correspondence between optimal parameters for weighted and unweighted cases and the resulting parameter setting rule hold with a high probability for a random weighted instance. 

As MaxCut is deeply connected to the SK model~\cite{Dembo2017}, we briefly discuss a ``weighted'' modification of the SK model obtained by drawing couplings in the SK model from $\mathcal{N}(\mu, \sigma^2)$ instead of $\mathcal{N}(0, 1)$. Here $\mu$ may depend on the problem size $N$. We call this modification ``biased SK'' and show that it behaves trivially in the infinite-size limit, unless $\mu=\mu(N)$ decays to zero with increasing $N$. %

We evaluate the parameter setting rule implied by Theorem~\introthmnum{} numerically outside of its theoretical assumptions by applying QAOA with $p\in \{1,2,3\}$ to MaxCut on a dataset of 34,701 weighted regular and non-regular graphs. We observe that our scheme outperforms the previously proposed approach of Ref.~\cite{Shaydulin2023}. On average, across all graphs, values of $p$ and edge-weight distributions, QAOA with parameters obtained using our scheme achieves solutions that are only 
$1.1$~percentage points (p.p.)~away from optimal, improving upon the $3.5$~p.p.~obtained using the technique presented in the prior work~\cite{Shaydulin2023}. Moreover, the disparity from the solutions obtained using optimized parameters is reduced by a factor of three (from $3.6$~p.p.~to $1.0$~p.p.)~when the edge weights are drawn from the exponential distribution, and by a factor of $\approx 6$ with the Cauchy distribution (from $20.7$ to $3.3$ p.p.).

We then propose a heuristic parameter rescaling rule for QAOA on arbitrary weighted problems. The heuristic rule is inspired by the theoretical results for the weighted MaxCut problem. As an example highlighting the generality of our observations, we consider QAOA applied to a portfolio optimization problem with a budget constraint, where the constraint is enforced throughout the QAOA evolution by the $\xgate\ygate$ Hamming-weight-preserving mixer. We observe that our simple rescaling procedure makes the landscape easier to optimize, reducing the number of iterations required for convergence to a fixed local optimum by a factor of 7.4x on a dataset of 280 portfolios with between 7 and 20 assets (qubits).

\section{Background}

We begin by briefly reviewing the Quantum Approximate Optimization Algorithm, the parameter setting schemes for it, and the weighted MaxCut problem. 

\subsection{Quantum Approximate Optimization Algorithm}

Consider the problem of optimizing some objective function $\mathcal{C}(\bm{x})$ defined on the $n$-dimensional Boolean cube that is encoded on $n$ qubits by a diagonal Hamiltonian $C = \text{diag}(\mathcal{C}(\bm{x}))$. Quantum Approximate Optimization Algorithm (QAOA)\cite{Hogg2000,farhi2014quantum} is a hybrid quantum-classical algorithm that approximately solves optimization problems by preparing a parameterized circuit such that upon measuring it, an approximate solution to the optimization problem is obtained. %
The QAOA circuit consists of layers of alternating unitaries, $e^{-i\gamma C}$ and $e^{-i\beta B}$, where $C$ is the Hamiltonian corresponding to the optimization problem and $B$ is the mixer Hamiltonian. Common choices of the mixer Hamiltonian $B$ include the transverse field ($B = \sum_j\xgate_j$) for unconstrained problems and the $\xgate\ygate$ mixer ($B = \frac{1}{2}\sum_{j,k}\left(\xgate_j\xgate_k + \ygate_j\ygate_k\right)$) for problems with an equality constraint on the Hamming weight.
The QAOA state with $p$ layers is given by
\begin{equation}
    \vert \bm\gamma, \bm\beta \rangle = e^{-i\beta_p B}e^{-i\gamma_p C}\ldots e^{-i\beta_1 B}e^{-i\gamma_1 C}|s\rangle,
\end{equation}
where $\vert s\rangle$ is the initial state and $\bm\gamma, \bm\beta$ are free parameters chosen by a classical routine. We discuss the strategies for setting the parameters $\bm\gamma, \bm\beta$ in Sec.~\ref{sec:parameter_setting_background} below.

The figure of merit that we use to evaluate the QAOA performance is the expected solution quality given by the  ``QAOA energy'':
\begin{equation}\label{eq:qaoa_obj}
    \langle C(\bm\gamma, \bm\beta)\rangle := \langle \bm\gamma, \bm\beta \vert C \vert \bm\gamma, \bm\beta \rangle = \sum_{\bm z\in\{0,1\}^{n}}\mathcal{C}(\bm z)\Pr(\bm z),
\end{equation}
where $\text{Pr}(\bm{z})$ is the probability of observing $\bm{z}$ when measuring all qubits of $\vert \bm\gamma, \bm\beta \rangle$.

\subsection{Parameter setting strategies for QAOA}\label{sec:parameter_setting_background}

Multiple techniques have been proposed for obtaining high-quality parameters for QAOA. While the parameters can be obtained by direct optimization of the objective \eqref{eq:qaoa_obj}
using a preferred optimization method~\cite{crooks2018performance,streif2020training, zhou2020quantum,shaydulin2019multistart,Lee2021, sack2021quantum, amosy2022iterative}, this procedure is typically computationally expensive~\cite{Lykov2020tensorqaoa,Medvidovic2021QAOA54qubit,Shaydulin2021Exploiting,Shaydulin2020CaseStudy}. The cost of finding parameters can be significantly reduced by leveraging the apparent problem-instance independence of the optimal QAOA parameters~\cite{brandao2018concentration,Akshay2021}. More straightforwardly, optimized parameters from one instance can be used directly as high-quality parameters for another instance from the same problem class~\cite{shaydulin2019multistart,Lotshaw2021BFGS,Galda2021transfer,lee2022depth,Shaydulin2023}. A machine learning model can be trained that would leverage the concentration to accurately predict the parameters~\cite{khairy2019learning,1907.05415,1911.04574,Wauters2020,Alam2020,Yao2021}. Optimal parameters can be derived exactly in certain analytically tractable cases, such as triangle-free regular graphs at $p=1$~\cite{Wang2018}.

In certain cases, i.e. in the infinite-size limit of a given problem, a closed-form iteration can be derived for the QAOA objective, Equation~\eqref{eq:qaoa_obj}, at constant $p$. Then parameters can be optimized in the infinite-size limit and used for finite-size instances. This has been demonstrated for the Sherrington-Kirkpatrick model~\cite{farhi2022quantum} and for MaxCut on random graphs~\cite{lipics.tqc.2022.7,2110.10685}. The goal of this work is to extend these results to weighted problems.

\subsection{MaxCut problem}
For an undirected graph $G = (V, E)$ with weights $w_{uv} = w_{\{u, v\}}$ assigned to edges $\{u, v\}\in E$, the goal of MaxCut is to partition the set of nodes $V$ into two disjoint subsets, such that the total sum of weights of the edges spanning both partitions is maximized. We refer to this problem as weighted MaxCut in the general case and as unweighted MaxCut when $w_{uv} = 1$ for all $\{u,v\}\in E$.

For the weighted MaxCut problem the objective function is given by
\begin{align}
   \mathcal{C}(\zv) = \frac{1}{2}\sum_{\{u,v\}\in E} w_{uv} (1 - z_u z_v),
\end{align}
where $z_u\in\{-1, 1\}$ are the variables to be optimized and $w_{uv}$ are sampled from the desired probability distribution. The MaxCut objective is encoded on qubits by the Hamiltonian
\begin{align}
   C = \frac{1}{2}\sum_{\{u,v\}\in E} w_{uv} (\idgate - \zgate_u \zgate_v),
\end{align}
where $\zgate_u$ and $\zgate_v$ are Pauli-Z operators applied to the $u$th and $v$th qubits, respectively.

For unweighted graphs, the cut fraction is defined as the ratio between the number of edges in a cut and the total number of edges in the graph. For a random unweighted $(D+1)$-regular graph, the optimal cut fraction is, with high probability, given by 
\begin{equation}\label{eq:cut_frac_opt}
    \frac{1}{2} + \frac{\Pi_*}{\sqrt{D}} + o\left(\frac{1}{\sqrt{D}}\right),
\end{equation}
where $\Pi_* \approx 0.7632$ is the Parisi value~\cite{Dembo2017}.

\section{Parameter setting scheme for QAOA on~weighted problems}\label{sec:scheme} %

Our parameter setting scheme is motivated by the observation, formalized in Sec.~\ref{sec:analytical}, that in many cases the QAOA energy landscape for weighted MaxCut can be rescaled to match that of unweighted MaxCut for arbitrary $p$. In the case of weighted MaxCut, this gives an explicit parameter setting rule. In the case of a general objective, we use the same observation to propose a rescaling rule that makes the QAOA energy landscape easier to optimize. We validate our scheme numerically for both cases in Section~\ref{sec:numerical}.

\subsection{Weighted MaxCut} 
The proposed procedure is as follows. First, rescale the edge weights in the graph following 
\begin{equation}
\label{eqn:general_maxcut_rule}
    w_{uv} \rightarrow \frac{w_{uv}}{\sqrt{\frac{1}{|E|}\sum_{\{u,v\}\in E}{w_{uv}^2}}}.
\end{equation}
Second, use the parameter setting rule for the corresponding unweighted graph. 

As an example of parameter setting rule for unweighted graphs to be used in the second step, one can use the parameters $\bm\beta^{\text{inf}}, \bm\gamma^{\text{inf}}$ optimized for large-girth regular graphs in the infinite-size limit~\cite[Tables 4 and 5]{lipics.tqc.2022.7} and follow the rescaling procedure therein, which we include here for completeness: $\bm\beta = \bm\beta^{\text{inf}}$, $\bm\gamma = \bm\gamma^{\text{inf}} / \sqrt{D}$. Here $D$ is the average degree of the graph. Alternatively, the procedure from Ref.~\cite{wurtz2021fixedangle} can be used. For small $D$ and $p$, higher quality results may be obtained by taking inspiration from the explicit formula of Ref.~\cite{Wang2018} and setting $\bm\gamma = \bm\gamma^{\text{inf}}\arctan{\frac{1}{\sqrt{D-1}}}$.

As an optional third step, the quality of the parameters can be improved further by running a local optimizer with a small initial step from the parameters obtained in the second step.

\subsection{General objective} 
For a general objective function and QAOA with an arbitrary mixer (e.g., constraint-preserving), analytical results are not available. %
At the same time, we can use the intuition from MaxCut to rescale the QAOA objective to make the geometry of the landscape more amenable to optimization. %
Specifically, if the objective $f$ is given by a degree-$k$ polynomial over spins $\bm{z} \in \{-1, 1\}^{n}$: 
\begin{equation}
    f(\bm{z}) = \sum_{\{u_1,\ldots, u_k\}} w^{(k)}_{u_1\ldots u_k}z_{u_1}\ldots z_{u_k} + \ldots + \sum_{u} w^{(1)}_u z_u,
\end{equation}
our first step is to divide the objective by 
\begin{align}
\label{eqn:gen_hypergraph_scaling_rule}
\sqrt{\frac{1}{\lvert E_{k}\rvert}\sum_{\{u_{1},\dots, u_{k}\}}(w^{(k)}_{u_1,\dots, u_k})^2 + \ldots + \frac{1}{\lvert E_1\rvert}\sum_{u}(w^{(1)}_{u})^2},\end{align}
where $E_{i}$ is the set of $i$-way hyperedges, i.e. the number of terms of order $i$.
In the second step, parameter optimization is performed as usual. 

This scaling is inspired by the observation that our results on weighted MaxCut generalize to problems with higher-order (higher than quadratic) terms; see Remark~\ref{rem:hypergraph_scale}. In Section~\ref{sec:numerical} we demonstrate the power of this simple procedure using the example of mean-variance portfolio optimization with a budget constraint enforced by the $\xgate\ygate$-mixer.

\section{Analytical results for QAOA on weighted MaxCut}\label{sec:analytical}

We now present the analytical results for QAOA applied to weighted MaxCut on large-girth regular graphs with i.i.d. edge weights. We begin by analyzing $p=1$ in Section~\ref{sec:qaoa_analytic_p_1}. QAOA energy for $p=1$ is given by a simple trigonometric formula derived in \cite[Theorem 7]{hadfield2018quantum}. We use this formula to derive globally-optimal QAOA parameters. The parameters we derive are optimal in expectation, with the expectation taken over the distribution of the edge weights. 
We first consider weights sampled from the exponential distribution and obtain optimal parameters for any graph size (Theorem \ref{thm:exponential_p_1}). We analyze the exponential distribution separately as it allows us to derive globally-optimal parameters for finite-sized graphs. Then we consider the infinite-size limit, which enables us to relax the assumption on the distribution and obtain optimal parameters for graphs with weights sampled from an arbitrary distribution (Theorem \ref{thm:p1_infinite_d}). 
We then consider $p\geq 1$ in Section~\ref{sec:arbitrary_distr_p_geq_1}. We extend the techniques of \cite{lipics.tqc.2022.7} to relate the QAOA objective landscape for weighted MaxCut to that for unweighted MaxCut (Theorem \ref{thm:gen_p_parameter_scaling}) and the SK model (Corollary \ref{cor:weighted_to_sk}). %

\subsection{Globally-optimal parameters for QAOA with $p=1$}
\label{sec:qaoa_analytic_p_1}
According to \cite[Theorem 7]{hadfield2018quantum}, the expected QAOA performance for MaxCut on triangle-free graphs can be expressed in closed form as:
\begin{align}
\label{eqn:weighted_maxcut_qaoa1}
    &\langle C(\gamma, \beta)\rangle = \frac{\sum_{\{u, v\} \in E} w_{uv}}{2} \nonumber\\&+ \frac{\sin(4\beta)}{4}\sum_{\{u,v\}\in E}w_{uv}\sin(w_{uv}\gamma)\left(\prod_{k \in \mathrm{nbhd}(u)/\{v\}}\cos(w_{uk}\gamma) + \prod_{t \in \mathrm{nbhd}(v)/\{u\}}\cos(w_{tv}\gamma)\right),
\end{align}
where $\mathrm{nbhd}(u)$ is the neighborhood function that gives the set of vertices adjacent to $u$.
The above is always maximized at $\beta=\frac{\pi}{8}$. Thus, with a slight abuse of notation, we define $\langle C(\gamma) \rangle = \langle C(\gamma, \frac{\pi}{8}) \rangle$.

We are considering the expected QAOA energy over the edge weights, i.e. $\mathbb{E}_{\mathbf{w}}[\langle C(\gamma)\rangle]$. In the sections that follow $\mathbb{E}_{\mathbf{w}}[\cdot]$ denotes the expectation over the graph weights, $w_{uv}$, that are all drawn independently from the distribution $\mathbf{w}$. Thus, for $(D+1)$-regular graphs with i.i.d. edge weights, this expectation simplifies as follows
\begin{align}
    \mathbb{E}_{\mathbf{w}}[\langle C(\gamma)\rangle]  &= \frac{\sum_{\{u, v\} \in E} \mathbb{E}_{\mathbf{w}} [w_{uv}]}{2} + \frac{1}{2} \sum_{\{u,v\} \in E}\mathbb{E}_{\mathbf{w}}\left[ w_{uv} \sin (w_{uv} \gamma)\right]
    \left(\mathbb{E}_{\mathbf{w}}\left[\cos(w_{uv}\gamma)\right]\right)^{D}\nonumber \\
    & = \frac{N(D+1)}{2}\left[\frac{\mathbb{E}_{\mathbf{w}}[w]}{2} + \frac{1}{2} \mathbb{E}_{\mathbf{w}}\left[ w \sin (w \gamma)\right]
    \left(\mathbb{E}_{\mathbf{w}}\left[\cos(w\gamma)\right]\right)^{D}\right], \label{eq:expected_weighted_maxcut_qaoa1}
\end{align}
where we drop the subscript on $w$ since the edge weights are i.i.d. and use the fact that
\begin{align}\label{eq:reg_graph_edge_count}
    |E| = \frac{N(D+1)}{2}.
\end{align}

We now consider edge weights distributed identically and independently according to the exponential distribution %
with parameter $\lambda > 0$, which has as its probability density function $f(x) = \lambda e^{-\lambda x}$ if $x>0$ or $f(x) = 0$ otherwise. The mean and standard deviation are $\mu = \sigma= \frac{1}{\lambda}$. 
\begin{theorem}[$p=1$, exponential distribution, finite size]
\label{thm:exponential_p_1}
Let $\mathbb{E}_{\text{exp}(\lambda)}[\langle C(\gamma)\rangle]$ denote the expected QAOA objective with $p=1$ over instances of the weighted MaxCut problem on a given triangle-free $(D+1)$-regular graph with edge weights, $w$, drawn i.i.d. from an exponential distribution with parameter $\lambda$. Then $\mathbb{E}_{\text{exp}(\lambda)}[\langle C(\gamma)\rangle]$ has a global maximum at $\gamma* = \frac{1}{\sqrt{\mathbb{E}_{\mathbf{w}}[w^2]}\sqrt{D+\frac{3}{2}}}$.
\end{theorem}

\begin{proof}
To obtain the optimal parameters, we start with Equation~\eqref{eq:expected_weighted_maxcut_qaoa1} and use the following identities
\begin{align}
    \mathbb{E}_{\mathbf{w}} [w\sin(w\gamma)] & = \frac{2\gamma\lambda^2}{(\lambda^2+\gamma^2)^2},\\
    \mathbb{E}_{\mathbf{w}}[\cos(w\gamma)] & = \frac{\lambda^2}{\lambda^2+\gamma^2},
\end{align}
which give
\begin{align}
    \mathbb{E}_{\text{exp}(\lambda)}[\langle C(\gamma)\rangle] & = \frac{N(D+1)}{2} \left[\frac{1}{2\lambda} + \frac{1}{2}\frac{2\gamma\lambda^2}{(\lambda^2+\gamma^2)^2}
    \left(\frac{\lambda^2}{\lambda^2+\gamma^2}\right)^{D}\right] \nonumber\\
    & = \frac{N(D+1)}{2}\left[\frac{1}{2\lambda} + \frac{\gamma \lambda^{2D+2}}{(\lambda^2+\gamma^2)^{D+2}} \right] \label{eq:exp_cost}.
\end{align}
Taking the derivative with respect to $\gamma$, we obtain
\begin{align}
    \frac{d}{d\gamma}\mathbb{E}_{\text{exp}(\lambda)}[\langle C(\gamma)\rangle] = c\frac{\lambda^2 - (2D+3)\gamma^2}{(\lambda^2 + \gamma^2)^{D+3}}, 
\end{align}
where $c$ is a positive and $\gamma$-independent constant. Setting the derivative to zero gives
\begin{equation}
    \gamma = \frac{\pm\lambda}{\sqrt{2D+3}} = \frac{\pm 1}{\sqrt{\mathbb{E}_{\exp(\lambda)}[w^2]}\sqrt{D+\frac{3}{2}}}.
\end{equation}
From Equation~\eqref{eq:exp_cost}, we can see that
\begin{align*}
    \mathbb{E}_{\text{exp}(\lambda)}[\langle C(-\infty)\rangle] = \mathbb{E}_{\text{exp}(\lambda)}[\langle C(\infty)\rangle] = \frac{N(D+1)}{4\lambda},
\end{align*} 
so the global maximum is 
\begin{equation}
    \gamma^* = \frac{1}{\sqrt{\mathbb{E}_{\exp(\lambda)}[w^2]}\sqrt{D+\frac{3}{2}}}.
\end{equation}
\end{proof}
Note that unlike the following Theorems, this proof does not rely on any assumptions on $D$. %

We now consider a graph with edge weights drawn from an arbitrary distribution with mean value $\mu$ and standard deviation $\sigma$. To study the infinite-size limit, we define a quantity that tends to a constant as $D\rightarrow\infty$. Specifically, we consider the following quantity
\begin{equation}
    \frac{\mathbb{E}_{\mathbf{w}}[\langle C(\gamma)\rangle]}{\mathbb{E}_{\mathbf{w}}\left[\sum_{\{u,v\}\in E}w_{uv}\right]},
\end{equation}
which reduces to the cut fraction if the graph is unweighted, i.e. $w_{uv}=1$; $\forall \{u, v\}\in E$. 

Using Equation~\eqref{eq:reg_graph_edge_count} followed by Equation~\eqref{eq:expected_weighted_maxcut_qaoa1}, we can write
\begin{align}
    \frac{\mathbb{E}_{\mathbf{w}}[\langle C(\gamma)\rangle]}{\mathbb{E}_{\mathbf{w}}\left[\sum_{\{u,v\}\in E}w_{uv}\right]}
    &=  \frac{2}{N(D+1)\mu} \mathbb{E}_{\mathbf{w}}\left[\langle C(\gamma)\rangle\right]\nonumber \\
    \label{eqn:p_1_trig}
    & \quad = \frac{1}{2} + \frac{1}{2\mu}\mathbb{E}_{\mathbf{w}}\left[ w \sin (w \gamma)\right]
    \left(\mathbb{E}_{\mathbf{w}}\left[\cos(w\gamma)\right]\right)^{D}  \\
    & \quad = \frac{1}{2} + \frac{\vartheta_{1}(D, \gamma)}{\sqrt{D}},
\end{align}
where we introduce $\vartheta_{1}(D, \gamma)$ to match $\Pi_*$ in Equation~\eqref{eq:cut_frac_opt}. We will now show that $\vartheta_{1}(D, \gamma)$ tends to a $D$-independent quantity as $D\rightarrow\infty$ when $\gamma=\Theta(D^{-1/2})$, and use the resulting limit to derive the optimal value $\gamma^*$ in the limit of infinite-sized graphs. %

The assumption of $\gamma= \Theta(D^{-1/2})$ is  inspired by the numerical observation that the optimal $\gamma$ for unweighted MaxCut is $\Theta(D^{-1/2})$ (see, e.g. \cite[Figure 1b]{2110.10685}). Furthermore, we prove that $\vartheta_{1}(D, \gamma)$ 
 has a local maximum at a value $\gamma=\Theta(D^{-1/2})$ for sufficiently large $D$. %
 In the limit of $D\rightarrow \infty$, we prove that this local maximum is also the global maximum. This motivates the definition of the following limiting quantity
 \begin{align}
\vartheta_1(\gamma) = \lim_{D\rightarrow\infty}\vartheta_1(D, \gamma/\sqrt{D}).
 \end{align}

 \begin{theorem}[$p=1$, infinite size]
 \label{thm:p1_infinite_d}
    Consider weighted MaxCut on a given triangle-free $(D+1)$-regular graph with edge weights, $w$, drawn i.i.d. from a distribution $\mathbf{w}$ with finite second moment. Then for sufficiently large $D$, the function $\vartheta_1(D, \gamma)$ associated with QAOA for $p=1$ has a local maximum at a $\gamma$ that is $\Theta(D^{-1/2})$. Moreover, the limiting quantity $\vartheta_1(\gamma)$  attains its global maximum at $\gamma^{*} = \frac{1}{\sqrt{\mathbb{E}_{\mathbf{w}}[w^2]}}$.
\end{theorem}
\begin{proof}
The assumption of finite second moment along with Jensen's inequality implies that $\mathbb{E}_{\mathbf{w}}{[\lvert w \rvert]}$ is also finite. Thus, since the derivatives of the functions inside the expectations taken in Equation~\eqref{eqn:p_1_trig} are dominated in $\gamma$, i.e. $\lvert w^2\cos(w\gamma)\rvert \leq w^2$ and  $\lvert w\sin(w\gamma)\rvert \leq \lvert w \rvert$, the dominated convergence theorem and the mean-value theorem \cite[Section 7.2.2]{glasserman2004monte} ensure that the operations of differentiation w.r.t. $\gamma$ and expectation over $w$ can be interchanged. This gives
\begin{align}
\frac{d}{d\gamma} \vartheta_{1}(D, \gamma) &=\frac{\sqrt{D}}{2\mu}(\mathbb{E}_{\mathbf{w}}[\cos(w\gamma)])^{D-1}\nonumber\\&\cdot\bigg(\mathbb{E}_{\mathbf{w}}[w^2\cos(w\gamma)]\mathbb{E}_{\mathbf{w}}[\cos(w\gamma)]- D(\mathbb{E}_{\mathbf{w}}[w\sin(w\gamma)])^{2}\bigg).
\end{align}
Substituting $\gamma = \frac{\alpha}{\sqrt{D}}$, for $\alpha$ independent of $D$ and using the Taylor series expansions of the trigonometric functions, we get
\begin{align}
    &\frac{d}{d\gamma}\vartheta_{1}(D, \gamma)\Bigr|_{\gamma = \alpha/\sqrt{D}} = \frac{\sqrt{D}}{2\mu}e^{-\frac{\mathbb{E}_{\mathbf{w}}[w^2]\alpha^2}{2}}[\mathbb{E}_{\mathbf{w}}[w^2] - (\mathbb{E}_{\mathbf{w}}[w^2])^2\alpha^2]+ O(D^{-1/2}),
\end{align}
where the implicit exchange of infinite series and  expectation over $w$ is justified by the finiteness of the second moment and Fubini's theorem \cite[Theorem 8.8]{rudin1974real}. Here we use the observation that
\begin{align}
\label{eqn:cosine_relation_exp}
\cos(x/\sqrt{D})^{D-1} &=  \left(1-\frac{x^2}{2D}\right)^{D-1} + O(D^{-1}) \nonumber\\
&=e^{-\frac{x^2}{2}} + O(D^{-1}),
\end{align} 
for $x$  that is bounded by a constant independent of $D$.
\footnote{Note that the second equality follows from $(1-\frac{x^2}{2D})^{D-1} = (1-\frac{x^2}{2D})^{D}(1-O(D^{-1}))= \left[\sum_{k=0}^{D}\frac{(-x^2/2)^{k}}{k!}\right](1 - O(D^{-1}))^{2}$ and works in the limit due to absolute convergence.}

For sufficiently large $D$, both \begin{align}
    \frac{d}{d\gamma}\vartheta_{1}(D, \gamma)\Bigr|_{\gamma = 0} = \frac{\sqrt{D}}{2\mu}\mathbb{E}_{\mathbf{w}}[w^2] + O(D^{-1/2}) > 0
\end{align} and 
\begin{align}
\frac{d}{d\gamma}\vartheta_{1}(D, \gamma)\Bigr|_{\gamma = \alpha^{*}/\sqrt{D}} \leq0
\end{align}
for some sufficiently large constant $\alpha^{*}$ independent of $D$. Thus, by Darboux's theorem \cite[Theorem 5.12]{rudin1976principles}, $\gamma$ has a local maxima in the interval $(0, \alpha^{*}/\sqrt{D})$ for each triangle-free $(D+1)$-regular graph. %

We now consider the limiting value of $\vartheta_{1}$ in the regime of small $\gamma$. With $\gamma = \frac{\gamma'}{\sqrt{D}}$ for some $D$-independent $\gamma'$, we get
\begin{align}
 \vartheta_{1}(D, \gamma'/\sqrt{D})
& = \frac{\sqrt{D}}{2\mu}\mathbb{E}_{\mathbf{w}} [w\sin(w\gamma'/\sqrt{D})](\mathbb{E}_{\mathbf{w}}[\cos(w\gamma'/\sqrt{D})])^{D} \nonumber\\ 
&= \frac{\mathbb{E}_{\mathbf{w}}[w^2]}{2\mu}\gamma'e^{-\frac{\mathbb{E}_{\mathbf{w}}[w^2]\gamma'^2}{2}} + O(D^{-1}),\nonumber
\end{align}
where we use Equation~\eqref{eqn:cosine_relation_exp}, and the implicit exchange of infinite series and expectation is justified by Fubini's theorem. Now taking the limit in $D$, we 
obtain 
\begin{align}
\vartheta_{1}(\gamma') := \lim_{D\rightarrow\infty}\vartheta_{1}(D, \gamma'/\sqrt{D}) = \frac{\mathbb{E}_{\mathbf{w}}[w^2]}{2\mu}\gamma'e^{-\frac{\mathbb{E}_{\mathbf{w}}[w^2]\gamma'^2}{2}}.
\end{align}
Now, consider the derivative, 
\begin{align}
\frac{d}{d\gamma'} \vartheta_{1}(\gamma') = \frac{\mathbb{E}_{\mathbf{w}}[w^2]e^{-\frac{\mathbb{E}_{\mathbf{w}}[w^2]\gamma'^2}{2}}}{2\mu}(1 - \mathbb{E}_{\mathbf{w}}[w^2]\gamma'^{2}).
\end{align}
It can be easily seen that the function $\vartheta_{1}(\gamma')$ is always decreasing to the right of the local maximum at $\gamma^{*}= \frac{1}{\sqrt{\mathbb{E}_{\mathbf{w}}[w^2]}}$, and the function is negative to the left of zero. Thus this is in fact a global optima.
\end{proof}

\begin{remark}
    To see the correspondence between Theorem~\ref{thm:p1_infinite_d} and Theorem~\ref{thm:exponential_p_1}, i.e. that the latter is a special case of the former, rescale $\gamma \to \gamma/\sqrt{D}$ and note that the constant in the denominator in Theorem~\ref{thm:exponential_p_1} has no effect on the limiting value as $D \to \infty$.
\end{remark}

\subsection{Concentration of QAOA objective at $p=1$}
We show that the QAOA objective for weighted MaxCut instances concentrates sharply around its expectation as $D \to \infty$, for triangle-free $(D+1)$-regular  graphs when $p=1$. This indicates that our scaling rules, which are derived by investigating the expectation of the objective, can also be expected to hold with high probability for \emph{any} weighted instance of a fixed graph. The QAOA objective in this setting is given in closed form in Equation~\eqref{eqn:weighted_maxcut_qaoa1}, where we set $\beta = \frac{\pi}{8}$. We will use the following bound on the partial derivatives of the objective as a function of the individual edge weights.

\begin{lemma}
\label{lem:partial-derivative-c}
Let $\langle C(\gamma) \rangle$ be the $p = 1$ QAOA objective for weighted MaxCut on triangle-free $(D+1)$-regular graphs as given in Equation~\eqref{eqn:weighted_maxcut_qaoa1}, with $\beta = \frac{\pi}{8}$. The partial derivatives of $\langle C(\gamma) \rangle$ with respect to the weights $w_{ij}$ satisfy,
\begin{align}
\left\lvert \frac{\partial \langle C(\gamma) \rangle}{\partial w_{ij}} \right\rvert &\le \frac{1}{2} + \frac{(1 + \gamma)\lvert w_{ij} \rvert}{2} + \frac{\gamma}{4} \bigg( \sum_{v \in \mathrm{nbhd}(i) / \{j\}} \lvert w_{iv} \rvert + \sum_{u \in \mathrm{nbhd}(j) /\{i\}} \lvert w_{uj} \rvert \bigg)
\end{align}
\end{lemma}
\begin{proof}
We define for convenience,
\begin{align}
T_{uv} = w_{uv}\sin(w_{uv}\gamma)\Bigg(\prod_{k \in \mathrm{nbhd}(u)/\{v\}}\cos(w_{uk}\gamma) \nonumber + \prod_{t \in \mathrm{nbhd}(v)/\{u\}}\cos(w_{vt}\gamma)\Bigg)
\end{align}
so that $\langle C(\gamma) \rangle = \sum_{\{u,v\} \in E} \frac{w_{uv}}{2} + \frac{T_{uv}}{4}$. Clearly, $\frac{\partial \langle C(\gamma) \rangle}{\partial w_{ij}} = \frac{1}{2} + \frac{1}{4}\sum_{\{u,v\} \in E} \frac{\partial T_{uv}}{\partial w_{ij}}$. To determine the last term, we have four cases,
\begin{enumerate}
    \item $u = i; v = j$: In $T_{uv}$, the term inside the parenthesis is upper bounded by 2 and independent of $w_{ij}$. We have, 
    \begin{align}
    \left\lvert \frac{\partial T_{uv}}{\partial w_{ij}} \right\rvert &\le 2(\sin(w_{ij}) + \gamma w_{ij} \cos(w_{ij})) 
    \\ &\le 2(1 + \gamma)\lvert w_{ij} \rvert.
    \end{align}

    \item $u = i; j \in \mathrm{nbhd}(u) / \{v\}$ : In $T_{uv}$, only the first term inside the parenthesis depends on $w_{ij}$. We have,
    \begin{align}
    \left\lvert \frac{\partial T_{uv}}{\partial w_{ij}} \right\rvert &\le \lvert w_{iv} \rvert \left( \gamma \cos(w_{ij})\right) \le \gamma \lvert w_{iv} \rvert.
    \end{align}

    \item $v=j; i \in \mathrm{nbhd}(v) /\{u\}$: In $T_{uv}$, only the second term inside the parenthesis depends on $w_{ij}$. We have, 
    \begin{align}
    \left\lvert \frac{\partial T_{uv}}{\partial w_{ij}} \right\rvert &\le \lvert w_{uj} \rvert \left( \gamma \cos(w_{ij})\right) \le \gamma \lvert w_{uj} \rvert.
    \end{align}

    \item $u \ne i, v \ne j$: $T_{uv}$ is independent of $w_{ij}$ and $\frac{\partial T_{uv}}{\partial w_{ij}} = 0$.
\end{enumerate}
Combining the above four cases, the result follows.

\end{proof}

We shall also use the following straightforward observation.
\begin{lemma}
\label{fact:n-ge-d}
    For any $(D + 1)$-regular triangle free graph, it must be that the number of vertices $N \ge D(D+1)$.
\end{lemma}
\begin{proof}
    Consider any node $u$. By definition $u$ has $D+1$ unique neighbors. Each neighbor of $u$ has $D$ neighbors in addition to $u$. Additionally, none of these new neighbors can also be neighbors of $u$ (or they will form a triangle). The observation follows.
\end{proof}

We are now ready to establish the concentration of $\langle C(\gamma) \rangle$ for some common edge weight distributions, at the value of $\gamma$ that optimizes the expected objective.

\subsubsection{Bounded Edge Weights}
In this setting the edge weights are independently and identically sampled from a distribution that is supported on an interval $[a,b]$, where $|a|,|b|$ are independent of $D$. We can show concentration by using the well-known McDiarmid's inequality~\cite{mcdiarmid_1989}, which we state in the necessary form below.

\begin{lemma}[McDiarmid's Inequality {[Adapted from \cite[Theorem~1]{warnke_2016}]}]
    \label{lem:mcdiarmid}
    Let $X_1,\dots,X_n$ be independent random variables each with range $\mathcal{X}$. Let $f \colon \mathcal{X}^{n} \to \mathbb{R}$ be any function with the \emph{bounded differences property}, i.e., for any co-ordinate $i \in [n]$ and points $(x_1,\dots,x_d),(x'_1,\dots,x'_d) \in \mathcal{X}^n$ such that $x_j = x'_j$ for all $j \in [n]\setminus\{i\}$, we have
    \begin{align}
        \label{eqn:bounded-difference}
        \lvert f(x_1,\dots,x_d) - f(x'_1,\dots,x'_d) \rvert \le c_i.
    \end{align}
    Then it holds that
    \begin{align}
        \label{eqn:mcdiarmid-concentration}
        \mathrm{Pr}[ \lvert f(X_1,\dots,X_n) - \mathbb{E}f(X_1,&\dots,X_n)\rvert \ge \epsilon] \le 2\exp\left(\frac{-2\epsilon^2}{\sum_{i=1}^{n} c_i^2}\right).     
    \end{align}
\end{lemma}

We now show a concentration inequality on the relative error from the mean of the QAOA objective under this setting.

\begin{theorem}[$p=1$, concentration, bounded weights]
    \label{thm:bounded-concentration-p1}
    Let the edge weights of a triangle-free $(D+1)$-regular  graph $(V,E)$ with $|V| = N$ vertices, be chosen from a distribution $\mathbf{w}$ with mean $\mu$, that is supported on the interval $[a,b]$, where $|a|,|b|,|\mu|$ are independent of $D$. Then, for all $\gamma = \Theta(D^{-\alpha})$ for some $\alpha \in (0,1)$, the QAOA objective $\langle C(\gamma) \rangle$ with $p=1$ satisfies the following concentration inequality, where $K$ is independent of $D$:
    \begin{align}
        &\mathrm{Pr}\left[\lvert \qaoaobjsingle{} - \mathbb{E}_{\mathbf{w}}[\qaoaobjsingle]\rvert \ge \epsilon\lvert\mathbb{E}_{\mathbf{w}}[\qaoaobjsingle]\rvert\right] \le 2\exp\left(\frac{-\epsilon^2 D^{1 + 2\alpha} (1 - O(D^{-\alpha}))}{K}\right).
    \end{align}
\end{theorem}
\begin{proof}
Assume wlog that $\lvert a\rvert \le \lvert b\rvert$.
The $N(D+1)/2$ edge weights $w_{ij}$ form a set of identical independent random variables taking values in $[a,b]$. Thus for any $\{i,j\} \in E$, $\lvert w_{ij} \rvert \le |b|$. Using the result of Lemma~\ref{lem:partial-derivative-c}, and noting that for any $\{i,j\} \in E$, there are exactly $D$ vertices in each of the sets $\mathrm{nbhd}(i)\setminus\{j\}$ and $\mathrm{nbhd}(j)\setminus\{i\}$, we have
\begin{align}
\label{eqn:bounded-pdev}
    \forall \{i,j\} \in E, \quad \left\lvert\frac{\partial \qaoaobjsingle}{\partial w_{ij}}\right\rvert \le  \frac{1 + (1+\gamma)|b| + \gamma D |b|}{2}.
\end{align}
The right hand side of the above inequality can be written as $K_1 + \gamma D K_2$ where $K_1,K_2$ are indepedent of $D$. 

Now, we view $\qaoaobjsingle$ as a function of $N(D+1)/2$ identical independent random variables.
By the mean value theorem, Equation~\eqref{eqn:bounded-pdev}, the assumption $\gamma = \Theta(D^{-\alpha})$, and the boundedness of $w_{ij}$, $\qaoaobjsingle$ satisfies the bounded differences property (Equation~\eqref{eqn:bounded-difference}) with respect to $w_{ij}$ with $c_{ij} \le \lvert b - a \rvert (K_1 + \gamma DK_2) \le K_3 \gamma D$ with $K_3$ independent of $D$. 

Noting that $\rvert\mathbb{E}_{\mathbf{w}}[\cos(w\gamma)]\lvert \le 1$ and $\mathbb{E}_{\mathbf{w}}[w\sin(w\gamma)] = O(D^{-\alpha})$, for $\gamma =\Theta(D^{-\alpha})$, from Equation~\eqref{eq:expected_weighted_maxcut_qaoa1} it follows that
\begin{align}
\lvert\mathbb{E}_{\mathbf{w}} [\qaoaobjsingle] \rvert \ge \frac{N(D+1)|\mu|(1 - O(D^{-\alpha}))}{2}.
\end{align}
Finally, %
\begin{align}
    \mathrm{Pr}\bigg[\lvert \qaoaobjsingle{} & - \mathbb{E}_{\mathbf{w}}\left[\qaoaobjsingle{}\right]\rvert \ge \epsilon\lvert\mathbb{E}_{\mathbf{w}}\left[\qaoaobjsingle\right]\rvert\bigg] \nonumber \\
    &\le 2\exp\left(\frac{-\epsilon^2 N^2 (D+1)^2 |\mu|^2 (1 - O(D^{-\alpha}))}{2\sum_{\{i,j\}\in E} c_{ij}^2}\right) \nonumber \\
    &\le 2\exp\left(\frac{-\epsilon^2 N (D+1) |\mu|^2(1 - O(D^{-\alpha}))}{K_3^2 \gamma^2 D^2}\right) \nonumber \\
    &\le 2\exp\left(\frac{-\epsilon^2 (D+1)^2 D^{2\alpha-1} |\mu|^2(1 - O(D^{-\alpha}))}{K_4^2}\right),
\end{align}
where $K_4$ is a constant independent of $D$.
Here, the first inequality follows from Lemma~\ref{lem:mcdiarmid}, and the final inequality follows from Lemma~\ref{fact:n-ge-d}. Setting $K = \frac{K_3^2}{|\mu|^2}$, the result of the theorem follows.

\end{proof}

Theorem~\ref{thm:bounded-concentration-p1} applies immediately to the case where the parameter $\gamma$ is set to $\frac{1}{\sqrt{\mathbb{E}_{\mathbf{w}}[w^2]D}}$.  Such a parameter setting is shown in Theorem~\ref{thm:p1_infinite_d} to maximize the expected objective in the limit $D \to \infty$.  

\subsubsection{Gaussian Edge Weights}

In order to show concentration when the weight distribution of each edge is Gaussian, we need an additional assumption that the number of vertices $N$ is such that $\log(N) = o(D\log(D))$. We use the following result on the sub-gaussian concentration of Lipschitz continuous functions of standard normal variables.

\begin{lemma}[Adapted from~{\cite[Theorem~5.2.2]{vershynin_2018}}]
\label{lem:gaussian-lipschitz}
    Consider an $n$-dimensional random vector $X$ drawn from the $n$-dimension standard normal, i.e. $X \sim \mathcal{N}(0,I_n)$, and a differentiable function $f \colon \mathbb{R}^n \to \mathbb{R}$ that is $L$-Lipschitz continuous, i.e. $\lVert \nabla f(x) \rVert_2 \le L$ for all $x \in \mathbb{R}^n$. Then,
    \begin{align}
    \mathrm{Pr}\left[\lvert f(X) - \mathbb{E}f(X) \rvert \ge t\right] \le \exp\left(\frac{-t^2}{K_{\psi}L^2}\right),
    \end{align}
    where $K_\psi$ is a universal constant.
\end{lemma}

The Lipschitz condition of Lemma~\ref{lem:gaussian-lipschitz} does not directly apply to the QAOA objective because the weights may be unbounded. Our argument will therefore address the concentration in two cases, using Lemma~\ref{lem:gaussian-lipschitz} to show concentration inside a suitably chosen interval, and Gaussian tail bounds to show that the probability of lying outside this interval is sufficiently small. We formalize this in the following concentration inequality on the relative error from the mean of the QAOA objective.

\begin{theorem}[$p=1$, concentration, Gaussian weights]
    \label{thm:gaussian-concentration}
    Let the edge weights of a $(D+1)$-regular triangle-free graph $(V,E)$ with $|V| = N$ vertices such that $\log(N) = o(D\log(D))$, be chosen from the normal distribution $\mathcal{N}(\mu,\sigma^2)$, where $\mu,\sigma$ are independent of $D$. Then, for all $\gamma = \Theta(D^{-1/2})$, the QAOA objective $\langle C(\gamma) \rangle$ with $p=1$ satisfies the following concentration inequality, where $K$ is independent of $D$:
    \begin{align}
        \mathrm{Pr}\bigg[\lvert \qaoaobjsingle{} - \mathbb{E}_{\mathbf{w}}[\qaoaobjsingle]\rvert \ge \epsilon\lvert\mathbb{E}_{\mathbf{w}}[\qaoaobjsingle]\rvert\bigg] \le 2\exp\left(\frac{-\epsilon D (1 - o(1))}{K \log(D)}\right).
    \end{align}
\end{theorem}
\begin{proof}
    For any $\epsilon > 0$, define the interval
    \begin{align}
        \mathcal{I}_\eta = [\mu - t_\eta\sigma, \mu +  t_\eta\sigma],
    \end{align}
    where \begin{align}
    t_\eta = \sqrt{2(\epsilon^{\eta} D\log(D) + \log(N) + \log(D+1))},
    \end{align} with $\eta$ to be chosen later. Let $\mathcal{B}_\eta$ denote the event that $w_{ij} \in \mathcal{I}_\eta$ for all $\{i,j\} \in E$, and the event $\mathcal{E}_\epsilon$ as $\left[\lvert \qaoaobjsingle{} - \mathbb{E}_{\mathbf{w}}[\qaoaobjsingle]\rvert \ge \epsilon\lvert\mathbb{E}_{\mathbf{w}}[\qaoaobjsingle]\rvert\right]$. From the definition of conditional probability,
    \begin{align}
    \label{eqn:split-tail-gaussian}
    \mathrm{Pr}[\mathcal{E}_\epsilon] \nonumber & = \mathrm{Pr}[\neg\mathcal{B}_\eta \cap \mathcal{E}_\epsilon] + \mathrm{Pr}[\mathcal{B}_\eta \cap \mathcal{E}_\epsilon] \nonumber\\
    &\le \mathrm{Pr}[\neg\mathcal{B}_\eta] + \mathrm{Pr}[\mathcal{E}_\epsilon \mid \mathcal{B}_\eta].
    \end{align}
    We bound $\mathrm{Pr}[\neg\mathcal{B}_\eta]$ by the following consideration:
    \begin{align}
        \label{eqn:bound-truncation-gaussian}
        \mathrm{Pr}[\neg\mathcal{B}_\eta] &= \mathrm{Pr}[\exists \{i,j\} \in E \text{ s.t. } w_{ij} \not\in \mathcal{I}_\eta] \nonumber \\
        &\le \frac{N(D+1)}{2} \exp(-t_{\eta}^2/2) \nonumber\\
        &\le \exp(-t_{\eta}^2/2 + \log(N) + \log(D+1))\nonumber \\
        &= \exp(-\epsilon^\eta D \log(D)),
    \end{align}
    where the first inequality follows by the union bound and the final equality by definition of $t_\eta$.
    
    To bound the second term on the RHS of \eqref{eqn:split-tail-gaussian}, note that in the event of  $\mathcal{B}_{\eta}$, the weights are all bounded in the interval $[\mu - t_\eta \sigma, \mu + t_\eta \sigma]$. The absolute value of any weight in this interval is bounded above by $b_\eta = \lvert \mu + \mathrm{sgn}(\mu) t_\eta \sigma \rvert$.
    We now consider the QAOA objective $\qaoaobjsingle$ as a function of a random $N(D+1)/2$-dimensional vector $Z = (z_{ij})_{\{i,j\} \in E}$ whose entries are each drawn from the standard normal distribution $\mathcal{N}(0,1)$ by defining $z_{ij} = \frac{w_{ij} - \mu}{\sigma}$. From Lemma~\ref{lem:partial-derivative-c} and an application of the chain rule,
    \begin{align}
\label{eqn:bounded-pdev-z}
    \forall \{i,j\} \in E, \quad \left\lvert\frac{\partial \qaoaobjsingle}{\partial z_{ij}}\right\rvert &= \left\lvert\frac{\partial w_{ij}}{\partial z_{ij}}\right\rvert \cdot \left\lvert\frac{\partial \qaoaobjsingle}{\partial w_{ij}}\right\rvert\nonumber\\
    & \le  \frac{\sigma(1 + (1+\gamma)b_\eta + \gamma D b_\eta)}{2} \nonumber\\
    &= \frac{\sigma \gamma D b_\eta (1 + O(D^{-1}))}{2}.
\end{align}
We can now bound the Lipschitz constant $L$ of $\qaoaobjsingle$ as
\begin{align}
    \label{eq:lipschitz-bound-gaussian}
    L \le \sqrt{\frac{N(D+1)}{2}} \max_{\{i,j\} \in E} \left\lvert\frac{\partial \qaoaobjsingle}{\partial z_{ij}}\right\rvert,
\end{align}
which leads to the following sequence of deductions:
\begin{align}
L^2 &\le \frac{(N D^2(D+1)\sigma^2 \gamma^2 \lvert \mu + \mathrm{sgn}(\mu) t_\eta \sigma \rvert^2)(1 + O(D^{-1}))}{8}\nonumber \\
&\le \frac{(N D^2(D+1)\sigma^4 \gamma^2 t_\eta^2)(1 + O(D^{-1}))}{8} \nonumber\\
&\le \frac{(N D^3(D+1) \log(D)\sigma^4 \gamma^2 \epsilon^{\eta})(1 + o(1))}{4} \nonumber\\
&\le K_1 N D^2(D+1) \log(D)\sigma^4 \epsilon^{\eta}(1 + o(1)),
\end{align}
where $K_1$ is a constant independent of $D$. The second inequality follows because $\frac{\mu}{\sigma t_\eta} = O(D^{-1/2})$, the third because $\log(D+1),\log(N) = o(D\log(D))$ by assumption,
and the fourth because $\gamma = \Theta(D^{-1/2})$ by assumption.

Finally, we can return to Equation~\eqref{eq:expected_weighted_maxcut_qaoa1} to see that
\begin{align}
\lvert\mathbb{E}_{\mathbf{w}} [\qaoaobjsingle] \rvert \ge \frac{N(D+1)|\mu|(1 - O(D^{-1/2}))}{2}.
\end{align}
Using Lemma~\ref{lem:gaussian-lipschitz} we have that,
\begin{align}
\label{eqn:bounded-concentration-gaussian}
\mathrm{Pr}[\mathcal{E}_\epsilon \mid \mathcal{B}_\eta] 
&\le \exp\left(-\frac{\epsilon^{2 - \eta}N(D+1)\mu^2(1 - o(1))}{4 K_\psi K_1 D^2 \log(D)\sigma^4 }\right)\nonumber \\
&\le \exp\left(-\frac{\epsilon^{2 - \eta}D(D+1)^2\mu^2(1 - o(1))}{4 K_\psi K_1 D^2 \log(D)\sigma^4 }\right)\nonumber \\
&\le \exp\left(-\frac{\epsilon^{2 - \eta}D(1 - o(1))}{K_2 \log(D)}\right),
\end{align}
where $K_2$ is a constant independent of $D$. Finally, setting $\eta = 1$, the result follows from Equations~(\ref{eqn:split-tail-gaussian},\ref{eqn:bound-truncation-gaussian},\ref{eqn:bounded-concentration-gaussian}).
\end{proof}

\subsubsection{Summary of concentration results}
In view of our concentration results for graphs with weight distributions that have bounded support (Theorem~\ref{thm:bounded-concentration-p1}), or for normally distributed weights (Theorem~\ref{thm:gaussian-concentration}), we see that, in the limit $D \to \infty$, and when $\gamma = \Theta(D^{-1/2})$, the typical relative deviations from the mean for the QAOA objective $\qaoaobjsingle$ are of the order $\tilde{O}(D^{-1})$, which ignores terms that are $O(\text{polylog}(D))$. Particularly, the probability of a relative deviation greater than $\epsilon$ are $\approx \exp(-\epsilon^2 D^2)$ and $\approx \exp(-\epsilon D / \log D)$ for bounded and Gaussian weights, respectively. Crucially, this shows that the relative dimensions in the quantity $\vartheta_1(D,\gamma)$ have relative deviations of the order $\tilde{O}(D^{-1/2})$. This is notable as $\vartheta_1(D,\gamma)$ tends to a constant as $D \to \infty$ when $\gamma = \Theta(D^{-1/2})$ (Theorem~\ref{thm:p1_infinite_d}), and is the primary quantity of interest when investigating the cut fraction (tending to the Parisi value in the unweighted case \cite{Dembo2017}).

\subsection{Correspondence between QAOA on weighted and unweighted graphs with $p\geq 1$}
\label{sec:arbitrary_distr_p_geq_1}

To derive a parameter scaling rule for arbitrary $p$, we extend the techniques developed in Ref.~\cite{lipics.tqc.2022.7} for MaxCut on large-girth, regular unweighted graphs to large-girth, regular graphs with i.i.d. edge weights. Without a subscript, $|\bm{\gamma}, \bm{\beta}\rangle$ will refer to the $p$-layer QAOA state for a random weighted instance of MaxCut on a given $(D+1)$-regular graph with weights, $w$, drawn from a distribution $\mathbf{w}$.

Note that
\begin{align}
\label{eqn:exp_weighted_maxcut_op}
\langle C(\bm\gamma, \bm\beta)\rangle =  \frac{1}{2}\sum_{\{u,v\}\in E}  w_{uv}(1 - \langle\bm{\gamma}, \bm{\beta}| \zgate_{u}\zgate_{v}|\bm{\gamma}, \bm{\beta} \rangle).
\end{align}
We start by proving the following result about the above quantity that is valid for any $p$.
\begin{lemma}
If a $(D+1)$ regular graph has girth $>2p+1$ and i.i.d. edge weights, $w$, drawn from $\mathbf{w}$, then the QAOA objective for weighted MaxCut on this graph satisfies
\begin{align}
\mathbb{E}_{\mathbf{w}}[\langle C(\bm\gamma, \bm\beta)\rangle] = \frac{N(D+1)\mu}{4} - \frac{N(D+1)}{4}\mathbb{E}_{\mathbf{w}}[w\langle\bm{\gamma}, \bm{\beta}| \zgate_{L}\zgate_{R}|\bm{\gamma}, \bm{\beta} \rangle],
\end{align}
for any edge $\{L, R\}$.
\end{lemma}
\begin{proof}
The locality of the $p$-layer QAOA combined with the regularity and girth $> 2p+1$ assumptions implies that the quantity $w_{LR}\langle\bm{\gamma}, \bm{\beta}| \zgate_{L}\zgate_{R}|\bm{\gamma}, \bm{\beta} \rangle$ can only depend on vertices that lie in a $D$-ary tree of diameter $2p+1$. This is the QAOA ``light cone'' for the term $\langle\bm{\gamma}, \bm{\beta}| \zgate_{L}\zgate_{R}|\bm{\gamma}, \bm{\beta} \rangle$. The tree subgraph, without weights, for $p=2$ and $D+1=4$ is shown in Figure \ref{fig:tree_subtraph}.
Thus, it follows from \cite[Equations A.2-A.5]{lipics.tqc.2022.7} that
\begin{align}
\label{eqn:zz_exp_expression}
w_{LR}\langle\bm{\gamma}, \bm{\beta}| \zgate_{L}\zgate_{R}|\bm{\gamma}, \bm{\beta} \rangle = \sum_{\{\bm{z}_{u}\}}w_{LR}z_{L}^{[0]}z_{R}^{[0]} \cdot\bigg[\exp\left(-i\sum_{\{u, v\}\in E_{LR}}w_{uv}\bm{\Gamma}\cdot (\bm{z}_{u}\bm{z}_{v})\right)\prod_{v \in V_{LR}}g(\bm{z}_{v})\bigg],
\end{align}
where $(V_{LR}, E_{LR})$ denotes the vertex and edge sets corresponding to the tree subgraph seen from the edge $\{L, R\}$. Note that Ref.~\cite{lipics.tqc.2022.7} has an extra factor of $\frac{1}{\sqrt{D}}$ in the exponential since the cost function in Equation~\eqref{eqn:exp_weighted_maxcut_op} is therein defined with that extra factor. In addition, $\bm{\Gamma}$ is a $(2p+1)$-component vector with entries $\bm{\Gamma}_{r} = \gamma_{r}$, $\bm{\Gamma}_{-r}=-\gamma_{r}$, $\bm{\Gamma}_{0} = 0$ and $1 \leq r \leq p$. Also for node $u$, the vector $\bm{z}_{u} = (z_u^{[1]}, \dots, z_u^{[r]}, z_u^{[0]},z_u^{[-r]}, \dots, z_u^{[-1]}) \in \{-1, 1\}^{2p+1}$. The notation $\{\bm{z}_{u}\}$ is short for the collection $\{\bm{z}_{u} | u \in V_{LR}, \bm{z}_{u} \in \{-1, 1\}^{2p+1}\}$. Lastly $(\bm{z}_{u}\bm{z}_v)$ denotes an element-wise product, and 
\begin{align}
g(\bm{z}_u) & = \left[\langle z_u^{[p]}|e^{i\beta_{p}\xgate}|z_u^{[0]}\rangle\langle z_u^{[0]}|e^{-i\beta_{p}\xgate}|z_u^{[-p]}\rangle\right]\nonumber\\&\cdot\frac{1}{2}\prod_{r=1}^{p-1} \langle z_u^{[-(r+1)]}|e^{-i\beta_{r}\xgate}|z_u^{[-r]}\rangle\langle z_u^{[r]}|e^{i\beta_{r}\xgate}|z_u^{[r+1]}\rangle.
\end{align}

\begin{figure}[t]
\centering
\includegraphics[width=0.47\textwidth]{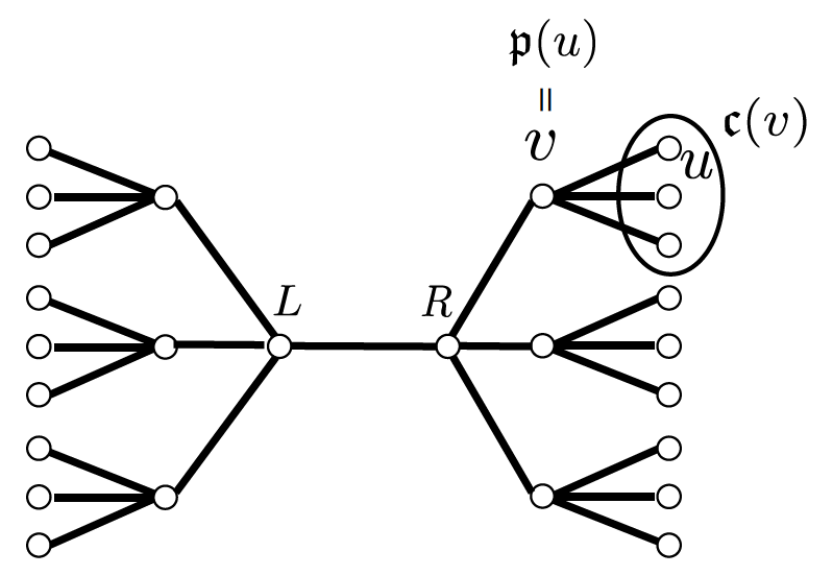}
\caption{\label{fig:tree_subtraph} Unweighted tree subgraph seen by QAOA from the edge $\{L, R\}$ with $p=2$ on a four-regular graph. The operation $\mathfrak{p}$ produces the parent of a node, and the operation $\mathfrak{c}$ produces the set of a node's immediate children.}
\end{figure}

The authors of \cite{lipics.tqc.2022.7} noted that Equation~\eqref{eqn:zz_exp_expression} can be computed recursively by traversing, from leaves to roots, the left and right branches simultaneously. This effectively ``factors'' the right-hand side of Equation~\eqref{eqn:zz_exp_expression}.  For simplicity we will only do this for $p=2$, however, the generalization is straightforward. Define $H_{D}^{(0)} = 1$. We start by summing over the configuration of an arbitrary leaf in either branch of Figure \ref{fig:tree_subtraph}:
\begin{align}
    \sum_{\bm{z}_u}g(\bm{z}_u)\exp\left(-iw_{u\mathfrak{p}(u)}\bm{\Gamma}\cdot(\bm{z}_{u}\bm{z}_{\mathfrak{p}(u)})\right),
\end{align}
where $\mathfrak{p}(u)$ is the parent of $u$ in the tree subgraph. Since the left and right branches are $D$-ary trees, the expression for any node $v$ in the second level of either branch is 

\begin{align}
    \sum_{\bm{z}_v}g(\bm{z}_v)H_{D}^{(1)}(\bm{z}_v)\exp\left(-iw_{v\mathfrak{p}(v)}\bm{\Gamma}\cdot(\bm{z}_{v}\bm{z}_{\mathfrak{p}(v)})\right),
\end{align}
where 
\begin{align}
    H^{(1)}_{D}(\bm{z}_{v}) &:=\prod_{u \in \mathfrak{c}(v)} \sum_{\bm{z}_u}H_{D}^{(0)}(\bm{z}_u)g(\bm{z}_u)\exp\left(-iw_{uv}\bm{\Gamma}\cdot(\bm{z}_{u}\bm{z}_{v})\right)\nonumber\\&=\prod_{u \in \mathfrak{c}(v)} \sum_{\bm{z}_u}H_{D}^{(0)}g(\bm{z}_u)\cos\left(w_{uv}\bm{\Gamma}\cdot(\bm{z}_{u}\bm{z}_{v})\right).
\end{align}
The second equality follows from the even parity of $g \cdot H^{(r)}_D$, implied by \cite[Claims A.14 and A.15]{lipics.tqc.2022.7} and $\mathfrak{c}(v)$ denotes the set of $D$ immediate children of $v$.

This can again be done for the next level, i.e. the roots $L$ or $R$, producing the quantity $H_{D}^{(2)}$. Lastly, we combine the results from the two branches by summing over the configurations of the left and right roots:
\begin{align}
\label{eqn:p_2_iteration}
w_{LR}\langle\bm{\gamma}, \bm{\beta}| \zgate_{u}\zgate_{v}|\bm{\gamma}, \bm{\beta} \rangle & = \sum_{\bm{z}_L,\bm{z}_{R}}w_{LR}z_{L}^{[0]}z_{R}^{[0]}\nonumber\\&\cdot\bigg[g(\bm{z}_{L})g(\bm{z}_{R})H_{D}^{(2)}(\bm{z}_L)H_{D}^{(2)}(\bm{z}_R)\exp\left(-iw_{LR}\bm{\Gamma}\cdot(\bm{z}_{L}\bm{z}_{R})\right)\bigg]\nonumber\\=&-i\sum_{\bm{z}_L,\bm{z}_{R}}z_{L}^{[0]}z_{R}^{[0]}g(\bm{z}_{L})g(\bm{z}_{R})\nonumber\\&\cdot\bigg[H_{D}^{(2)}(\bm{z}_L)H_{D}^{(2)}(\bm{z}_R)w_{LR}\sin\left(w_{LR}\bm{\Gamma}\cdot(\bm{z}_{L}\bm{z}_{R})\right)\bigg],
\end{align}
where the last equality follows again from the even parity of $g \cdot H^{(r)}_D$. The general iteration for the random quantity $w_{LR}\langle\bm{\gamma}, \bm{\beta}| \zgate_{L}\zgate_{R}|\bm{\gamma}, \bm{\beta} \rangle$ follows by induction.

Note that in the previous recursion, each edge was only counted once. Since all of the edges are i.i.d., the expectation operation commutes with all products that appear in the right-hand side of Equation~\eqref{eqn:p_2_iteration}. More specifically, we have the following for general $p$:
\begin{align}
\label{eqn:gen_exp_zz}
\mathbb{E}_{\mathbf{w}}[w_{LR}\langle\bm{\gamma}, \bm{\beta}| \zgate_{L}\zgate_{R}|\bm{\gamma}, \bm{\beta} \rangle] 
& =-i\sum_{\bm{z}_L,\bm{z}_{R}}z_{L}^{[0]}z_{R}^{[0]}g(\bm{z}_{L})g(\bm{z}_{R})\nonumber\\&\cdot\bigg[\bar{H}_{D}^{(p)}(\bm{z}_L)\bar{H}_{D}^{(p)}(\bm{z}_R)\mathbb{E}_{\mathbf{w}}[w\sin\left(w\bm{\Gamma}\cdot(\bm{z}_{L}\bm{z}_{R})\right)]\bigg],
\end{align}
where subscripts have been dropped from the weights due to the i.i.d. assumption, and
\begin{align}
\label{eqn:gen_exp_h}
\bar{H}^{(r)}_{D}(\bm{z}_{v}):=\mathbb{E}_{\mathbf{w}}[H^{(r)}_{D}(\bm{z}_{v})] =\prod_{u \in \mathfrak{c}(v)} \sum_{\bm{z}_u}\bar{H}_{D}^{(r-1)}(\bm{z}_u)g(\bm{z}_u)\mathbb{E}_{\mathbf{w}}[\cos\left(w\bm{\Gamma}\cdot(\bm{z}_{u}\bm{z}_{v})]\right)
\end{align}
for $1\leq r \leq p$.

Since all randomness has been removed by the expectation operation, it is evident that this quantity only depends on the graph structure and not the sampled weights, since they are i.i.d. The non-random quantities in the iteration for $w_{LR}\langle\bm{\gamma}, \bm{\beta}| \zgate_{L}\zgate_{R}|\bm{\gamma}, \bm{\beta} \rangle$ only depend on the local graph structure that QAOA sees from a given edge. As argued earlier this graph structure is always two $D$-ary trees joined at their roots. Thus, the quantity $\mathbb{E}_{\mathbf{w}}[w\langle\bm{\gamma}, \bm{\beta}| \zgate_{L}\zgate_{R}|\bm{\gamma}, \bm{\beta} \rangle]$
is independent of the chosen edge $\{L, R\}$, which is analogous to the unweighted case. 

The result of the lemma follows from Equation~\eqref{eqn:exp_weighted_maxcut_op} 
and the linearity of expectation.
\end{proof}

By the previous lemma, we can, analogously to the $p=1$ case, define $\vartheta_p(D,\bm{\gamma}, \bm{\beta})$ as follows:
\begin{align}
 \frac{\mathbb{E}_{\mathbf{w}}[\langle C(\bm{\gamma}, \bm{\beta})\rangle]}{\mathbb{E}_{\mathbf{w}}\left[ \sum_{\{u,v\}\in E}w_{uv}\right]} & =  \frac{\mathbb{E}_{\mathbf{w}}[\langle C(\bm{\gamma}, \bm{\beta})\rangle]}{\frac{N(D+1)\mu}{2}} \nonumber \\
&=\frac{1}{2} - \frac{1}{2\mu}\mathbb{E}_{\mathbf{w}}[w\langle\bm{\gamma}, \bm{\beta}| \zgate\zgate|\bm{\gamma}, \bm{\beta} \rangle]\nonumber\\
&= \frac{1}{2} + \frac{\vartheta_{p}(D, \bm{\gamma}, \bm{\beta})}{\sqrt{D}},
\end{align}
and
\begin{align}
\vartheta_{p}(\bm{\gamma}, \bm{\beta}) & := \lim_{D\rightarrow\infty}\vartheta_{p}(D, \bm{\gamma}/\sqrt{D}, \bm{\beta})\nonumber\\&=\lim_{D\rightarrow\infty}-\frac{\sqrt{D}}{2\mu}\mathbb{E}_{\mathbf{w}}\bigg[w\bigg\langle\frac{\bm{\gamma}}{\sqrt{D}}, \bm{\beta}\bigg| \zgate\zgate\bigg|\frac{\bm{\gamma}}{\sqrt{D}},\bm{\beta}\bigg\rangle\bigg],
\end{align}
where the subscripts on $w$ and $\zgate$ have been dropped since they can be arbitrary by the previous lemma.

The quantity considered in \cite[Equation A.19]{lipics.tqc.2022.7} for the unweighted case is the following

\begin{align}
&\nu_{p}(\bm{\gamma}, \bm{\beta})=\lim_{D\rightarrow\infty}-\frac{\sqrt{D}}{2}_\text{u}\bigg\langle\frac{\bm{\gamma}}{\sqrt{D}}, \bm{\beta}\bigg| \zgate\zgate\bigg|\frac{\bm{\gamma}}{\sqrt{D}},\bm{\beta}\bigg\rangle_{\text{u}},
\end{align}
where the subscript ``u'' indicates that the parameterized state is prepared by a $p$-layer QAOA for the corresponding unweighted problem on the same graph. Our main result below shows that $\bm{\gamma}$ can be scaled to make these two quantities equal up to a global scaling factor.

\begin{theorem}[$p\geq 1$, infinite size]
\label{thm:gen_p_parameter_scaling}
If the girth $>2p+1$, and the edge-weight distribution, $\mathbf{w}$, has finite second moment, then for  all parameters $\bm{\gamma}, \bm{\beta}$ the following holds
\begin{align}
\label{eqn:zz_unweighted_relation_to__weighted}
\nu_{p}(\bm{\gamma}, \bm{\beta}) = \frac{\mu}{\sqrt{\mathbb{E}_{\mathbf{w}}[w^2]}}\vartheta_{p}\left(\frac{\bm{\gamma}}{\sqrt{\mathbb{E}_{\mathbf{w}}[w^2]}}, \bm{\beta}\right).
\end{align}
\end{theorem}
\begin{proof}
We implicitly assume that $\bm{\gamma} \mapsto \frac{\bm{\gamma}}{\sqrt{D}}$ and thus $\bm{\Gamma} \mapsto \frac{\bm{\Gamma}}{\sqrt{D}}$ in Equation~\eqref{eqn:gen_exp_zz}.
By the product rule for limits, we can evaluate the limits of the terms $\bar{H}_{D}^{(p)}(\bm{z}_{L})$, $\bar{H}_{D}^{(p)}(\bm{z}_{R})$, and the one involving the $\sin$ separately, since we will show they individually exist. 
Note that for any sum inside of the product of Equation~\eqref{eqn:gen_exp_h}:
\begin{align}
\label{eqn:limiting_cosine}
\sum_{\bm{z}_u}g(\bm{z}_u)\bar{H}_{D}^{(p-1)}(\bm{z}_{u})&\mathbb{E}_{\mathbf{w}}\left[\cos\left(w\frac{\bm{\Gamma}}{\sqrt{D}}\cdot(\bm{z}_{u}\bm{z}_{v})\right)\right]\nonumber\\&=\bigg(1 - \frac{\mathbb{E}_{\mathbf{w}}[w^2]}{2D}\sum_{\bm{z}_u}g(\bm{z}_u)\bar{H}_{D}^{(p-1)}(\bm{z}_{u})(\bm{\Gamma}\cdot(\bm{z}_{u}\bm{z}_{v}))^{2}  + O(D^{-2})\bigg),
\end{align}
where the implicit exchange of the expectation operator and infinite series expansion of trig functions is justified by Fubini's theorem and the assumption of the weight distribution having finite second moment, like in Section \ref{sec:qaoa_analytic_p_1}. In addition, we have used the following generalization of \cite[Equation (A.23)]{lipics.tqc.2022.7}, where for any $r$: 
\begin{align}
\sum_{\bm{z}_u}g(\bm{z}_u)H_{D}^{(r-1)}(\bm{z}_{u}) = 1.
\end{align}
After taking expectations, it follows that $
\sum_{\bm{z}_u}g(\bm{z}_u)\bar{H}_{D}^{(r-1)}(\bm{z}_{u}) = 1$. By the i.i.d. assumption, Equation~\eqref{eqn:limiting_cosine} is the same for every $u \in \mathfrak{c}(v)$, and thus 
\begin{align}
\bar{H}_{D}^{(p)}(\bm{z}_{v}) & =\bigg(1 - \frac{\mathbb{E}_{\mathbf{w}}[w^2]}{2D}\sum_{\bm{z}_u}g(\bm{z}_u)\bar{H}_{D}^{(p-1)}(\bm{z}_{u})(\bm{\Gamma}\cdot(\bm{z}_{u}\bm{z}_{v}))^{2}\bigg)^{D}\nonumber + O(D^{-1})\nonumber\\
&=\exp\left(-\frac{\mathbb{E}[w^2]}{2}\sum_{\bm{z}_{u}}g(\bm{z}_{u})\bar{H}_{D}^{(p-1)}(\bm{z}_{u})(\bm{\Gamma}\cdot(\bm{z}_u\bm{z}_{v}))^2\right)\nonumber +O(D^{-1}),
\end{align}
where we use Equation~\eqref{eqn:cosine_relation_exp}.
Along with continuity, the previous result implies
\begin{align}
\bar{H}^{(p)}(\bm{z}_{v}) &:=\lim_{D\rightarrow\infty}\bar{H}_{D}^{(p)}(\bm{z}_{v}) \nonumber\\&=\exp\left(-\frac{\mathbb{E}_{\mathbf{w}}[w^2]}{2}\sum_{\bm{z}_{u}}g(\bm{z}_{u})\bar{H}^{(p-1)}(\bm{z}_{u})(\bm{\Gamma}\cdot(\bm{z}_u\bm{z}_{v}))^2\right).
\end{align}
The limit can then be propagated down to the lowest level of the recursion.
Similarly for the term involving: %
\begin{align}
\lim_{D\rightarrow\infty}\sqrt{D}\mathbb{E}_{\mathbf{w}}\left[w\sin\left(w\frac{\bm{\Gamma}}{\sqrt{D}}\cdot(\bm{z}_L\bm{z}_R)\right)\right] = \mathbb{E}_{\mathbf{w}}[w^2]\bm{\Gamma}\cdot(\bm{z}_L\bm{z}_R)
\end{align}
Putting this altogether, for arbitrary $\{L, R\} \in E$, we have

\begin{align}
\vartheta_{p}(\bm{\gamma}, \bm{\beta}) = \frac{i\sqrt{\mathbb{E}_{\mathbf{w}}[w^2]}}{2\mu}\sum_{\bm{z}_L,\bm{z}_{R}}z_{L}^{[0]}z_{R}^{[0]}g(\bm{z}_{L})g(\bm{z}_{R}) \bigg[\bar{H}^{(p)}(\bm{z}_L)\bar{H}^{(p)}(\bm{z}_R)\sqrt{\mathbb{E}_{\mathbf{w}}[w^2]}\bm{\Gamma}\cdot(\bm{z}_L\bm{z}_R)\bigg]
\end{align}
and comparing with the unweighted case \cite[Equation A.26]{lipics.tqc.2022.7} reveals that the equality in Equation~\eqref{eqn:zz_unweighted_relation_to__weighted} holds.
\end{proof}

\begin{remark}
Note that in the i.i.d. case the square of the denominator of Equation~\eqref{eqn:general_maxcut_rule} is an unbiased estimator of $\mathbb{E}_{\mathbf{w}}[w^2]$. The i.i.d. case is the inspiration for the more general scaling rule.
\end{remark}

\begin{remark}
\label{rem:hypergraph_scale}
By \cite[Section 6]{lipics.tqc.2022.7}, which extends the iteration presented earlier for unweighted regular graphs to unweighted regular $k$-uniform hypergraphs, Theorem \ref{thm:gen_p_parameter_scaling} can be trivially extended to weighted regular $k$-uniform hypergraphs.
\end{remark}

This result implies a relationship in the infinite-size limit between QAOA's objective value for weighted MaxCut and the SK model. Let 
\begin{align}
\label{eqn:sk_objective}
    C^{\text{SK}} = \sum_{1\leq u < v\leq N}\frac{J_{uv}}{\sqrt{N}}\zgate_u\zgate_v,
\end{align}
where $J_{uv} \sim \mathcal{N}(0, 1)$,
and 
\begin{align}
\label{eqn:sk_limiting}
    V_{p}(\bm{\gamma}, \bm{\beta}) := \lim_{n\rightarrow\infty}\mathbb{E}_{J}[{}_{J}\langle\bm{\gamma},\bm{\beta}| C^{\text{SK}}/N|\bm{\gamma},\bm{\beta}\rangle_{J}],
\end{align}
where the subscript $J$ of the state signifies that the state was prepared by a $p$-layer QAOA with the SK objective as the phase operator.

\begin{theorem}[Restated from \cite{lipics.tqc.2022.7}]
For all $p$ and all parameters $(\bm{\gamma}, \bm{\beta})$ the following holds
\begin{align}
    V_{p}(\bm{\gamma}, \bm{\beta}) = \nu_{p}(\bm{\gamma}, \bm{\beta}).
\end{align}
\end{theorem}
Trivially, this in combination with Theorem \ref{thm:gen_p_parameter_scaling} leads to the following corollary.
\begin{corollary}
\label{cor:weighted_to_sk}
If the edge-weight distribution has finite second moment, then for all $p$ and all parameters $(\bm{\gamma}, \bm{\beta})$ the following holds
\begin{align}V_{p}(\bm{\gamma}, \bm{\beta}) = \frac{\mu}{\sqrt{\mathbb{E}_{\mathbf{w}}[w^2]}}\vartheta_{p}\left(\frac{\bm{\gamma}}{\sqrt{\mathbb{E}_{\mathbf{w}}[w^2]}}, \bm{\beta}\right).
\end{align}
\end{corollary}
Thus the performance of QAOA on SK, MaxCut  on large-girth, regular graphs and weighted MaxCut on large-girth, regular graphs are equivalent in the {infinite-size} limit. 

\begin{remark}
By \cite[Theorem 3]{Basso_2022} and Remark \ref{rem:hypergraph_scale}, one can trivially extend Corollary \ref{cor:weighted_to_sk} to  connect QAOA's performance on weighted MaxCut on regular $k$-uniform hypergraphs to its performance on pure $k$-spin models, generalizing SK.
\end{remark}

\section{Observations about biased SK model}
\label{sec:SK_mu_neq_0}
As presented in Corollary \ref{cor:weighted_to_sk}, there is a deep connection between arbitrarily-weighted MaxCut and the SK model. The SK model is given in Equation~\eqref{eqn:sk_objective} and has couplings $J_{uv} \sim \mathcal{N}(0, 1)$. A natural generalization to consider is a model which has couplings $J_{uv} \sim \mathcal{N}(\mu, \sigma^2)$ with arbitrary $\mu$ and $\sigma$. More generally, we can allow for the bias to be a function of the number of spins, i.e. $\mu(N)$. When $\mu(N) \neq 0$, we call this the biased SK model, and when $\mu(N) = 0$, we call it the standard SK model. Unfortunately, this natural generalization does not lead to interesting behavior. Specifically, we show that unless $\mu(N)\rightarrow 0$, %
the biased SK problem is trivial in  the thermodynamic limit.

The performance of QAOA for arbitrary $p$ on standard SK, specifically an iteration for the quantity $V_{p}$, was originally established in \cite{farhi2022quantum} using different techniques than those of Section \ref{sec:arbitrary_distr_p_geq_1}. However, it is not clear how these techniques can be generalized to non-symmetric distributions. In this section, we use a different set of elementary techniques to determine the limiting optimal value of different versions of the biased SK model. Our goal is to find an analog to the Parisi value \cite{parisi1980sequence, talagrand2006parisi} for the biased model. %
The following is based on Ref.~\cite{panchenko2014introduction}.

For $\bm{z} \in \{-1, 1\}^{N}$, let 
\begin{align}
    G(\bm{z}) = -\sum_{1\leq i <j\leq N}J^{(N)}_{ij}z_iz_j,
\end{align} where $J^{(N)}_{ij} \sim \mathcal{N}(\mu(N), \sigma^2)$. The optimization problem is $\max_{\bm{z}}G(\bm{z})$. Note that unlike the standard SK, which is symmetric around zero, here we must keep track of the signs of the couplings. %

When $\mu(N) =0$, we know that $\mathbb{E}_{J}[\max_{\bm{z}}G(\bm{z})] = O(N^{3/2})$. In the standard SK model, the weights are scaled by $N^{-1/2}$ to ensure that expected maximum instead grows linearly with $N$. This is the reasoning for the scaling of the standard SK objective presented in Equation~\eqref{eqn:sk_objective}. For even $N$, let

\begin{align}
\label{eqn:h_function}
h(N) = \begin{cases}
    \frac{N^2}{4}, & \text{if $\mu(N) > 0$},\\
    -\binom{N}{2}, & \text{if $\mu(N) < 0$},\\
    0, & \text{if $\mu(N) = 0$,}
    \end{cases}
\end{align}
which equals $\frac{\max_{\bm{z}}\mathbb{E}_{J}[G(\bm{z})]}{\mu(N)}$ when $\mu(N)\neq 0$. More specifically, when $\mu(N) > 0$, the problem $\max_{\bm{z}}\mathbb{E}_{J}[G(\bm{z})]$ reduces to MaxCut on a complete graph with all edge weights equal to $\mu(N)$, and thus the optimal cut value is $\mu(N)\frac{N^2}{4}$, i.e. set half of the $z_i=1$. When $\mu(N) < 0$, the optimal value is obtained when all $z_i = 1$, and results in an objective function value of $\mu(N)\binom{N}{2}$. Note that when $N$ is odd, the factor in the denominator is the same for all cases, and thus we can restrict to  even $N$, wlog.

We have  the following by the convexity of $\max$:
\begin{align}
    \mu(N)h(N)= \max_{\bm{z}}\mathbb{E}_{J}[G(\bm{z})] \leq \mathbb{E}_{J}[\max_{\bm{z}}G(\bm{z})]
\end{align}
and
\begin{align}
 \mathbb{E}_{J}[\max_{\bm{z}}G(\bm{z})] & = \frac{1}{\alpha}\log\exp\mathbb{E}_{J}[\alpha\max_{\bm{z}}G(\bm{z})] \nonumber \\ 
 &\leq \frac{1}{\alpha}\log\mathbb{E}_{J}[\exp(\alpha\max_{\bm{z}}G(\bm{z}))] \nonumber\\
 &\leq \frac{1}{\alpha}\log \sum_{\bm{z}}\mathbb{E}_{J}[\exp(G(\bm{z}))] \nonumber\\
 & = \frac{1}{\alpha}\log \sum_{\bm{z}}\exp\left(\frac{\alpha^{2}\sigma^2}{2}\binom{N}{2} - \alpha\mu(N)\sum_{i<j}z_iz_j\right) \nonumber \\ &\leq \frac{1}{\alpha}\log2^{N}\exp\left(\frac{\alpha^2\sigma^2}{2}\binom{N}{2}  + \alpha\mu(N)h(N)\right) \nonumber\\ 
 &= \frac{N\log(2)}{\alpha} + \frac{\alpha\sigma^2}{2}\binom{N}{2} + \mu(N)h(N).
\end{align}
Minimizing over $\alpha > 0$, gives

\begin{align}
    \alpha^2 = \frac{N\log(4)}{\sigma^2 \binom{N}{2}},
\end{align}
and thus
\begin{align}
\label{eqn:exp_max_bound}
    \mu(N)h(N) \leq \mathbb{E}_{J}[\max_{\bm{z}}G(\bm{z})] \leq \sigma\sqrt{\log(4)N\binom{N}{2}} + \mu(N)h(N).
\end{align}
Note, the right-hand side of the last inequality now involves only $\sigma$ instead of $\sigma^2$, and thus is invariant under any scaling before or after the expectation in $\mathbb{E}_{J}[\max_{\bm{z}}G(\bm{z})]$. 

{When $\mu(N)=0$, we have $h(N)=0$, and we recover the SK scaling, i.e.,} 
\begin{align}
\mathbb{E}_{J}[\max_{\bm{z}}G(\bm{z})] \leq \sigma\sqrt{\log(4)N\binom{N}{2}} = O(N^{3/2}).
\end{align}
{Also for standard SK with $J_{ij} \sim \mathcal{N}(0, 1)$, which is a special case of $\mu(N)=0$, it is known that }
\begin{align}
\lim_{N\rightarrow\infty}\frac{\mathbb{E}_{J}[\max_{\bm{z}}G(\bm{z})]}{N^{3/2}} = \Pi^{*},
\end{align}
{where $\Pi^{*}$ is the Parisi value} \cite{panchenko2014introduction}. {Thus for $J_{ij} \sim \mathcal{N}(0, \sigma^2)$, it is easily seen that}
\begin{align}
    \lim_{N\rightarrow\infty}\frac{\mathbb{E}_{J}[\max_{\bm{z}}G(\bm{z})]}{N^{3/2}} = \sigma\Pi^{*}.
\end{align}

{To summarize}, the limiting behavior for SK with a distribution with any $\sigma \neq 1$ can be obtained by simple rescaling. 
{To connect this to the quantities discussed at the end of  Section \ref{sec:arbitrary_distr_p_geq_1}}, for $\mathcal{N}(0, \sigma^2)$, the quantity $V_{p}$ in Equation~\eqref{eqn:sk_limiting} scales as $V_{p}(\bm{\gamma}, \bm{\beta}) = V_{p, 1}(\bm{\gamma}, \bm{\beta}) =V_{p, \sigma}(\frac{\bm{\gamma}}{\sigma}, \bm{\beta})$. Therefore for the remainder of this Section we focus on the case where $\mu(N)\neq 0$.

If $\mu(N)\rightarrow \mu$ as $N\rightarrow\infty$ for
some nonzero constant $\mu$, the term involving $\mu(N)$ dominates, and the expected maximum is $\Theta(N^2)$. {More specifically, using the definition of $h(N)$ in \eqref{eqn:h_function} for nonzero $\mu$, we get that $\lim_{N\rightarrow\infty}\mu(N)h(N) = \Theta(N^2)$, and both sides of Equation~\eqref{eqn:exp_max_bound} are of the same order in $N$.}
Thus {the squeeze theorem implies} that the limiting quantity is exactly
\begin{align}
\lim_{N\rightarrow\infty}\frac{\mathbb{E}_{J}[\max_{\bm{z}}G(\bm{z})]}{N^2} = \frac{\mu}{2(1+\text{sign}(\mu))},
\end{align}
where we define $\text{sign}(x)$ to be $1$ when $x\geq0$ and $0$ when $x<0$. {This accounts for the constant factor differences in $h(N)$ for $\mu \neq 0$.} The solution to the above problem is trivial in the infinite limit: set all $z_i = 1$, when $\mu < 0$, and  set half of the $z_i$ to $1$ when $\mu > 0$.

However, to compensate for the maximum growing $\Theta(\sqrt{N})$ faster {(than standard SK)} when $\mu(N)$ is not always zero, we could consider defining the biased SK model to have $\mu(N) = O(N^{-1/2})$. Then letting $\mu = \lim_{N\rightarrow\infty}\mu(N) \sqrt{N}$ we get, from Equation~\eqref{eqn:exp_max_bound},

\begin{align}
0 \leq   \lim_{N\rightarrow\infty}\frac{\mathbb{E}_{J}[\max_{\bm{z}}G(\bm{z})]}{N^{3/2}} - \frac{\mu}{2(1+\text{sign}(\mu))} \leq \frac{\sigma}{\sqrt{2}} \log^{1/2}(4).
\end{align}
{Thus the growth of the maximum is again $O(N^{3/2})$, like in the original SK. Recall that the original SK satisfies}
\begin{align}
0 \leq   \lim_{N\rightarrow\infty}\frac{\mathbb{E}_{J}[\max_{\bm{z}}G(\bm{z})]}{N^{3/2}} \leq \frac{\sigma}{\sqrt{2}} \log^{1/2}(4),
\end{align}
{with the limiting value being the Parisi value.}

In this regime, the ``biased SK'' model appears non-trivial and {the limiting value and solution are not obvious}. %
{We are unable to find any mention of such a model in existing literature,} and we leave the study of its properties to future work.

\begin{figure*}[t]
\includegraphics[width=1\textwidth]{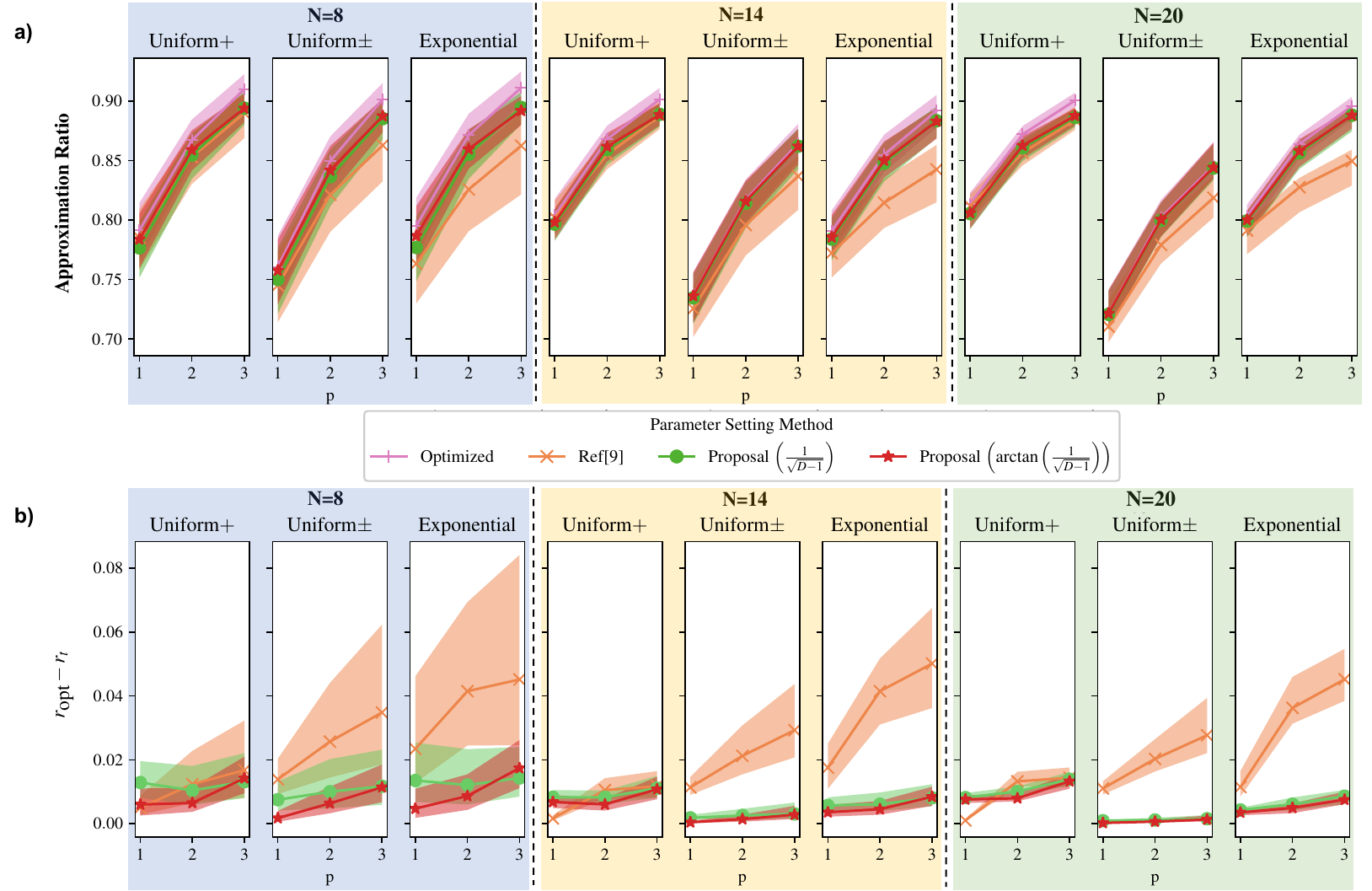}
\caption{\label{fig:transferred_approx_ratio} a) The approximation ratios obtained with directly optimized parameters, parameter setting method of Ref.~\cite{Shaydulin2023}, and parameter setting methods presented in this work. b) The gap between the approximation ratios with optimized parameters and with parameter setting methods of Ref.~\cite{Shaydulin2023}, ($\lowerromannumeral{1}$) and ($\lowerromannumeral{2}$). 
The proposed parameter setting methods perform better when compared to the prior work,
as indicated by the reduced gap from the objective obtained with the optimized parameters.}
\end{figure*}

\section{Numerical results}\label{sec:numerical}

\subsection{Weighted MaxCut}
Numerical investigation of the proposed parameter setting rule has been performed on a dataset of weighted graphs from Ref. \cite{Shaydulin2023}, 
available through QAOAKit~\cite{shaydulin2021qaoakit}. The dataset consists of a total of 34,701 weighted graphs with up to 20 nodes and contains both regular and non-regular graphs. The graphs have edge weights drawn i.i.d. from four different distributions, namely Uniform over $[0,1]$ (``Uniform$+$''), Uniform over $[-1,1]$ (``Uniform $\pm$''), Exponential (with $\lambda=0.2$), and Cauchy.

For the numerical study, we investigate two proposed parameter setting rules, which are variants of Equation~\eqref{eqn:general_maxcut_rule}:
\begin{equation}
    \begin{aligned}\setcounter{subeqns}{0}
        \subnumber\quad &\bm{\gamma}* = \frac{\bm{\gamma}^{\text{inf}}}{\sqrt{\frac{D-1}{|E|}\sum_{\{u,v\}\in E}{w_{uv}^2}}},
    \end{aligned}
    \label{eq:sqrt_D}
\end{equation}
\begin{equation}
    \begin{aligned}
        \subnumber\quad &\bm{\gamma}* = \frac{\bm{\gamma}^{\text{inf}}}{\sqrt{\frac{1}{|E|}\sum_{\{u,v\}\in E}{w_{uv}^2}}} \arctan\Big(\frac{1}{\sqrt{D-1}}\Big),
    \end{aligned}
    \label{eq:arctan_sqrt_D}
\end{equation}
where the parameters $\bm\beta^{\text{inf}}, \bm\gamma^{\text{inf}}$ are the optimized parameters for large-girth, regular graphs in infinite-size limit from \cite[Table 4]{qaoa_high_girth_arxiv_ver}, and $D$ is the average degree.
 Our baseline is the parameter scheme of Ref. \cite{Shaydulin2023}, given by:
 \begin{align}
     \bm{\gamma}* = \frac{\bm{\gamma}^{\text{median}}}{\frac{1}{|E|}\sum_{\{u,v\}\in E}|w_{uv}|}\arctan(\frac{1}{\sqrt{D-1}}),
 \end{align}
 where $\bm{\gamma}^{\text{median}}$ is a median taken over optimized parameters for all $261,080$ nonisomorphic connected 9-node graphs. The key difference between our scaling and that of Ref.~\cite{Shaydulin2023} is the choice of the denominator. Since $\bm{\gamma}^{\text{median}}$ is close in value to $\bm{\gamma}^{\text{inf}}$, the nominator is similar in both schemes.

We refer to the parameter setting procedure described in Equation~\eqref{eq:sqrt_D} as method ($\lowerromannumeral{1}$) and that in Equation~\eqref{eq:arctan_sqrt_D} as method ($\lowerromannumeral{2}$). %
The first method is inspired directly by the analytical results described in Section~\ref{sec:analytical}. 
We observe that in the case of small $p$, better results are obtained when the formula for $p=1$ from \cite{Wang2018} is considered, which motivates the second rule. %
We note that while we have derived the exact formula for graphs with weights sampled from the exponential distribution, we do not use it in the numerical experiments. Our goal for numerics is to simulate the practical setting, wherein one does not know the distribution from which the weights are sampled.

We analyze the performance of the proposed parameter setting rules across multiple weight distributions, values of $p$ and values of $N$. %
Herein, we denote the median approximation ratio with directly optimized parameters by $r_{opt}$, with the parameter setting scheme from Ref. \cite{Shaydulin2023} by $r_{[9]}$ and with the two proposed methods as $r_{D}$ and $r_{\arctan}$ respectively. We refer to the difference between the approximation ratio of a given parameter setting scheme and $r_{opt}$ as the optimality gap. The results are presented in Figure~\ref{fig:transferred_approx_ratio}. 

Our techniques lead to lower optimality gaps as compared to Ref.~\cite{Shaydulin2023} in all cases except $p=1$ with  weights sampled uniformly from $[0,1]$. We note that the gap between the methods ($\lowerromannumeral{1}$) and ($\lowerromannumeral{2}$) reduces as $p$ increases. For example, for $N=8$ the optimality gap drops from $0.0111$ on average for $p=1$, to $0.0062$ for $p=2$, and eventually to $0.0005$ for $p=3$.

The median difference in approximation ratios for all considered $p$ and weight distributions is 1.8 p.p. for method ($\lowerromannumeral{1}$) and 1.45 p.p. for method ($\lowerromannumeral{2}$).  Specifically, for the cases of exponential and Cauchy distributions, the median differences in approximation ratios from our method ($\lowerromannumeral{1}$) are 1.3 p.p.~and 3.8 p.p.~respectively, and those from method ($\lowerromannumeral{2}$) remain a mere 1.0 p.p.~and 3.3 p.p.~respectively. For comparison, the previous proposal~\cite{Shaydulin2023} obtains  median differences of 3.6 p.p.~and 20.7 p.p.~for the weights drawn from exponential and Cauchy distributions respectively.  As can be seen in Figure~\ref{fig:transferred_score_cauchy}, for the case when the edge weights are drawn from a Cauchy distribution, the improvement over Ref.~\cite{Shaydulin2023} is the largest, with an $8\times$ reduction in optimality gap at $p=3$.

\begin{table*}[t]
    \centering
    \begin{tabular}{ c|c|c|c|c }
        {Weight Distribution} & $p$ & $r_{opt}-r_{[9]}$ & $r_{opt}-r_{D}$ & $r_{opt}-r_{\arctan}$ \\
        \hhline{=|=|=|=|=}
        \multirow{3}{4em}{Uniform$+$} & $1$ & $0.5$ & $1.3$ & $0.6$\\ \cline{2-5}
        & $2$ & $1.2$ & $1.0$ & $0.6$ \\ \cline{2-5}
        & $3$ & $1.6$ & $1.3$ & $1.4$  \\
        \hline
        \multirow{3}{4em}{Uniform$\pm$} & $1$ & $1.4$ & $0.7$ & $0.2$ \\ \cline{2-5}
        & $2$ & $2.6$ & $1.1$ & $0.6$ \\ \cline{2-5}
        & $3$ & $3.5$ & $1.3$ & $1.1$ \\
        \hline
        \multirow{3}{4em}{Exponential} & $1$ & $2.3$ & $1.3$ & $0.5$ \\ \cline{2-5}
        & $2$ & $4.1$ & $1.2$ & $0.8$ \\ \cline{2-5}
        & $3$ & $4.5$ & $1.4$ & $1.7$ \\
        \hline
        \multirow{3}{4em}{Cauchy} & $1$ & $17.5$ & $4.8$ & $3.3$ \\ \cline{2-5}
        & $2$ & $20.8$ & $4.1$ & $3.9$ \\ \cline{2-5}
        & $3$ & $24.0$ & $3.2$ & $3.0$ \\
        \hline
        \hline
    \end{tabular}
    \caption{The optimality gaps achieved by the proposed parameter setting rules and the rule in~\cite{Shaydulin2023}. Proposed method leads to lower optimality gaps in all cases except $p=1$ for the Uniform$+$ distribution.}
    \label{tab:median_aggregate_approx_ratios}
\end{table*}
\begin{table*}[t]
    \centering
    \begin{tabular}{ p{2cm}|c|c|c|c }
        \Centering{$s^2$} & $p$ & $r_{opt}-r_{[9]}$ & $r_{opt}-r_{D}$ & $r_{opt}-r_{\arctan}$ \\
        \hhline{=|=|=|=|=}
        \multirow{3}{4em}{\Centering{0.5697 [Uniform$+$]}} & $1$ & $0.00166$ & $0.00828$ & $0.0068$\\ \cline{2-5}
        & $2$ & $0.01053$ & $0.00827$ & $0.00604$ \\ \cline{2-5}
        & $3$ & $0.01172$ & $0.01112$ & $0.01077$  \\
        \hline
        \multirow{3}{4em}{\Centering{0.9557 [Exponential]}} & $1$ & $0.00572$ & $0.00265$ & $0.00357$ \\ \cline{2-5}
        & $2$ & $0.04149$ & $0.00611$ & $0.00455$ \\ \cline{2-5}
        & $3$ & $0.05012$ & $0.0081$ & $0.00843$ \\
        \hline
        \multirow{3}{4em}{\Centering{1.1364 [Uniform$\pm$]}} & $1$ & $0.00183$ & $0.00160$ & $0.00045$ \\ \cline{2-5}
        & $2$ & $0.02124$ & $0.00258$ & $0.00142$ \\ \cline{2-5}
        & $3$ & $0.02928$ & $0.00306$ & $0.00282$ \\
        \hline
        \multirow{3}{4em}{\Centering{2.4811 [Cauchy]}} & $1$ & $0.17524$ & $0.04775$ & $0.03344$ \\ \cline{2-5}
        & $2$ & $0.20788$ & $0.04098$ & $0.03867$ \\ \cline{2-5}
        & $3$ & $0.23958$ & $0.03238$ & $0.02967$ \\
        \hline 
        \hline
    \end{tabular}
    \caption{The median quality of solution for median variance across different distributions and for multiple $p$ achieved through our introduced methods when compared with the previous study\cite{Shaydulin2023} for $n=14$ . As variance, $s^2$, grows, the optimality gaps achieved by our presented techniques are almost an order of magnitude better in comparison to the method in Ref. \cite{Shaydulin2023}.}
    \label{tab:approx_ratios_var}
\end{table*}
\begin{figure}[H]
\centering
\includegraphics[width=0.7\textwidth]{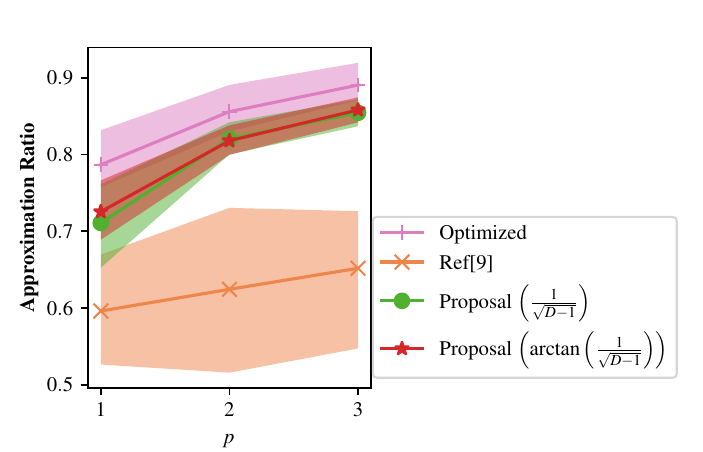}
\caption{\label{fig:transferred_score_cauchy} Approximation ratio for the graphs with edge weights drawn from a Cauchy distribution for $N=14$. The proximity to the optimized parameter scenario, especially for large $p$, indicates the power of the suggested parameter setting strategies and shows a clear improvement over the earlier work. Our methods reduce the optimality gap by a factor of 8 for $p=3$ as compared to Ref.~\cite{Shaydulin2023}.}
\end{figure}

\begin{figure*}[t]
    \centering
    \includegraphics[width=\textwidth]{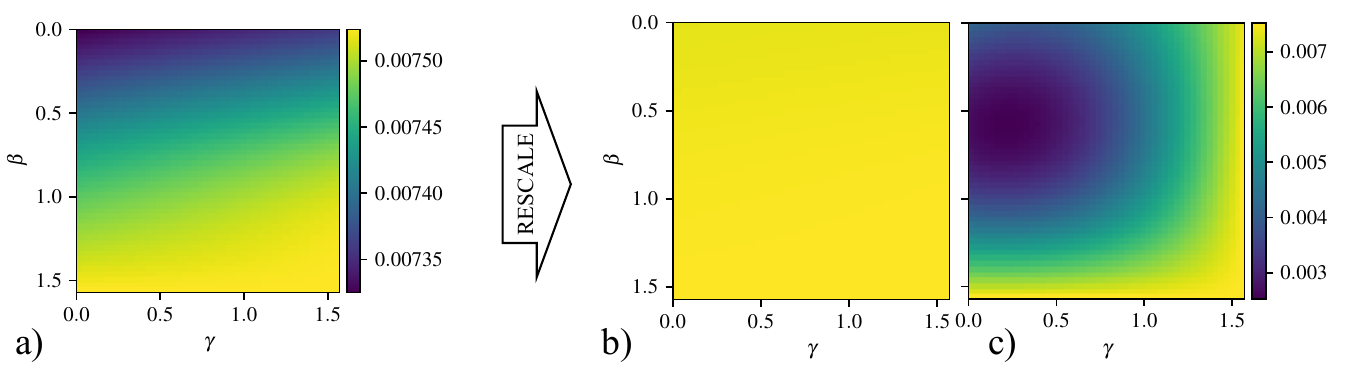}
    \caption{Energy landscape of $\xgate\ygate$-QAOA applied to a random portfolio optimization problem with (a,b) original and (c) rescaled objective. Rescaling improves the geometry of the landscape. To highlight the flatness of the unrescaled landscape, (b) plots the same landscape as (a), but using the color scheme of (c). %
    }
    \label{fig:heatmaps}
\end{figure*}

\begin{figure}[h]
    \centering
    \includegraphics[width=0.6\textwidth]{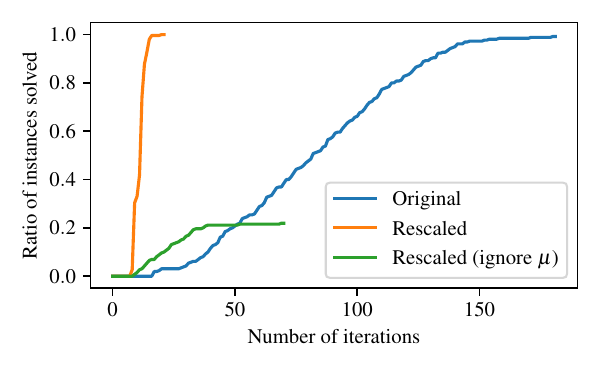}
    \caption{Performance profiles for BOBYQA run from a fixed initial point on 280 instances between 7 and 20 qubits. A line plots the ratio of instances solved with a given number of iterations. Rescaling the objective reduces the number of iterations until convergence to the target local optimum by 7.4x on average. {Ignoring weights for linear terms, i.e. $\mu$, leads to incorrect rescaling and to optimizer failing to recover the optimum.}}
    \label{fig:traces}
\end{figure}

The optimality gaps obtained by considering our presented parameters setting rules are of comparable values for each $p$. However, the performance of the method of Ref. \cite{Shaydulin2023} deteriorates as $p$ increases. %

From the values shown in Table. \ref{tab:approx_ratios_var}, it can be observed that when the variance is small, the performance obtained using the method in the prior work~\cite{Shaydulin2023} is comparable to that of our methods. However, as variance increases the solution qualities achieved by the methods introduced in this work beat those of the previous work~\cite{Shaydulin2023}.

\subsection{Portfolio optimization}

While the analytical results of Section~\ref{sec:analytical} only apply to MaxCut and QAOA with transverse-field mixer $B=\sum_j\xgate_j$, the intuition applies more broadly. To illustrate this point, in this Section we consider a portfolio optimization problem with a budget constraint, given by
\begin{align}
    & \min_{\bm{x}\in \{0, 1\}^N} q\bm{x}^{\mathsf{T}}\Sigma\bm{x}-\bm{\mu}^{\mathsf{T}}\bm{x}, \\
    &\text{s.t.} \sum_j x_j = k.
\end{align}
This problem is commonly considered as a target for QAOA~\cite{herman2022portfolio,brandhofer2022benchmarking,slate2021quantum,hodson2019portfolio,Hao2022, he2023alignment}. We use random instances generated by \texttt{RandomDataProvider} in \texttt{qiskit\_finance}~\cite{qiskitfinance} and set $q=0.5$ and $k = \lfloor \frac{n}{2}\rfloor$. We consider 20 instances for each number of qubits between 7 and 20, for a total of 280 instances. For each instance, we reformulate the problem in terms of spin variables and rescale the objective
following Equation~\eqref{eqn:gen_hypergraph_scaling_rule}, with the scaling coefficient
computed separately for each instance. We use QAOA with $p=1$ and the $\xgate\ygate$ mixer given by $B = \frac{1}{2}\sum_{j=1}^N \left(\xgate_{j}\xgate_{j+1}+\ygate_j\ygate_{j+1}\right)$ and set the initial state $\ket{s}$ to be the uniform superposition over all Hamming weight $k$ states (Dicke state). We denote this variant $\xgate\ygate$-QAOA.

We first observe that the rescaling improves the optimization landscape. The QAOA energy at $p=1$ is presented in Figure~\ref{fig:heatmaps}. The landscape with the original objective function (Figure~\ref{fig:heatmaps}a,b) is flat, making the optimization difficult. After rescaling (Figure~\ref{fig:heatmaps}c), a clear local minimum is visible. Moreover, the geometry of the landscape is similar with regards to both parameters. This suggests that the optimization of QAOA parameters with the rescaled objective should be easier.

To quantify this improvement, we optimize the parameters using the NLopt~\cite{NLOpt} implementation of the BOBYQA~\cite{powell2009bobyqa} gradient-free optimizer. We use BOBYQA as it has been shown to be an effective optimizer for QAOA parameters~\cite{shaydulin2019multistart}. For each instance, we run one local optimization starting from the optimal QAOA parameters for the Sherrington-Kirkpatrick model obtained from \cite[Table 4]{qaoa_high_girth_arxiv_ver} and given by $\gamma^{\text{init}} = 1$, $\beta^{\text{init}}=\frac{\pi}{4}$\footnote{The parameters in \cite[Table 4]{lipics.tqc.2022.7} differ from these by a factor of 2. The difference is due to the constant factors in the QAOA simulator implementation that we use.}. We set the stopping criteria to be $\text{\texttt{xtol}}=\text{\texttt{ftol}}=10^{-8}$. We present the performance profiles in Figure~\ref{fig:traces}. The optimizer finds the same optimum for the original and rescaled objective in $91\%$ of the cases. We consider that an optimizer ``solves'' the instance when it recovers this optimum. When the objective is rescaled, the optimizer takes on average $7.4$ times fewer iterations to find the same optimum.

{To contrast portfolio optimization with the weighted MaxCut problem and to highlight the importance of including the weight for the linear term in Eq.~\ref{eqn:gen_hypergraph_scaling_rule}, we additionally plot the performance profile for results with rescaling rule where the coefficients for linear terms $\mu_j$ are ignored during the rescaling and only $\Sigma_{ij}$ are used. This profile is marked ``ignore $\mu$'' in Figure~\ref{fig:traces}. As can be seen from the plot, ignoring weights for linear terms leads to incorrect rescaling, with the optimizer failing to recover the optimum in the majority of cases. We note that the weights $\mu_j$ act analogously to bias in the ``biased SK'' model of Section~\ref{sec:SK_mu_neq_0} and can dominate the objective for some instances.}

\section{Discussion}

In this work, we propose heuristic parameter setting rules for QAOA, inspired by a formal connection between weighted and unweighted MaxCut on regular graphs. For $p=1$, we derive explicit expressions for the parameter $\gamma$ that maximizes the cost function in the weighted case. Our analysis of MaxCut at $p=1$ rigorously proves that the globally-optimal $\gamma$ are small, providing additional justification for this commonly used assumption~\cite{lipics.tqc.2022.7,boulebnane2023peptide,2208.06909,2110.10685}. For $p\ge 1$, we show explicitly how the energy landscape and, consequently, the optimal parameters  scale between the weighted and unweighted cases. As we prove the concentration of the QAOA objective, our results apply with high probability to any random weighted MaxCut instance. {An important limitation of our analysis is the high-girth assumption, which we inherit from Ref.~\cite{qaoa_high_girth_arxiv_ver}}. Surprisingly, the simple rules we derive for MaxCut {on high-girth, regular graphs} apply broadly, which we demonstrate by extensive numerical experiments on {MaxCut on a general class of random graphs and on} random instances of a constrained portfolio optimization problem. %

Additionally, we consider the biased SK problem and rigorously show that it has a trivial solution in the infinite-size limit, unless the mean of the weight distribution falls sufficiently fast with the number of vertices. This investigation was inspired by the connection between SK and MaxCut on regular graphs, and the observation that the closed-form iterations that we use for QAOA do not apply to complete graphs. However, it appears that, unlike standard SK, the analysis of QAOA performance is unlikely to lead to significant insights when the weights are biased.

Our observation that QAOA parameters $\bm\gamma$ have to decrease with problem size is an instantiation of a broader principle, namely that parameterized quantum circuits are not scale-independent. Similar results have been observed for quantum kernel methods~\cite{Shaydulin2022,2206.06686} and quantum neural network initialization~\cite{zhang2022escaping}. A unification of these observations into a general theory of parameterized quantum circuits is a tempting prospect, though it would require the development of novel mathematical techniques.

\section*{Acknowledgements}
The authors thank Kunal Marwaha for early exploratory discussions, Aram Harrow for insightful discussions on biased SK, and their colleagues at the Global Technology Applied Research center of JPMorgan Chase for support and helpful discussions. 
JB was supported by a grant from the Simons Foundation under Award No. 825053.

\bibliographystyle{quantum}
\bibliography{research}

\begin{thebibliography}{10}

\bibitem{nielsen2010quantum}
Michael~A Nielsen and Isaac~L Chuang.
\newblock ``Quantum computation and quantum information''.
\newblock \href{https://dx.doi.org/10.1017/CBO9780511976667}{Cambridge
  university press}. ~(2010).

\bibitem{herman2022}
Dylan Herman, Cody Googin, Xiaoyuan Liu, Alexey Galda, Ilya Safro, Yue Sun,
  Marco Pistoia, and Yuri Alexeev.
\newblock ``A survey of quantum computing for finance''~(2022).
\newblock  url:~\url{https://doi.org/10.48550/arXiv.2201.02773}.

\bibitem{Hogg2000}
Tad Hogg and Dmitriy Portnov.
\newblock ``Quantum optimization''.
\newblock \href{https://dx.doi.org/10.1016/s0020-0255(00)00052-9}{Information
  Sciences {\bf 128}, 181--197}~(2000).

\bibitem{farhi2014quantum}
Edward Farhi, Jeffrey Goldstone, and Sam Gutmann.
\newblock ``A quantum approximate optimization algorithm''~(2014).
\newblock  url:~\url{https://doi.org/10.48550/arXiv.1411.4028}.

\bibitem{hadfield2019quantum}
Stuart Hadfield, Zhihui Wang, Bryan O’Gorman, Eleanor~G Rieffel, Davide
  Venturelli, and Rupak Biswas.
\newblock ``From the quantum approximate optimization algorithm to a quantum
  alternating operator ansatz''.
\newblock Algorithms {\bf 12}, 34~(2019).
\newblock  url:~\url{https://doi.org/10.3390/a12020034}.

\bibitem{2208.06909}
Sami Boulebnane and Ashley Montanaro.
\newblock ``Solving boolean satisfiability problems with the quantum
  approximate optimization algorithm''~(2022).
\newblock  url:~\url{https://doi.org/10.48550/arXiv.2208.06909}.

\bibitem{lipics.tqc.2022.7}
Joao Basso, Edward Farhi, Kunal Marwaha, Benjamin Villalonga, and Leo Zhou.
\newblock ``The quantum approximate optimization algorithm at high depth for
  maxcut on large-girth regular graphs and the sherrington-kirkpatrick model''.
\newblock \href{https://dx.doi.org/10.4230/LIPICS.TQC.2022.7}{Proceedings of
  the Conference on the Theory of Quantum Computation, Communication and
  Cryptography {\bf 7}, 1--21}~(2022).

\bibitem{2111.12641}
Matthew~B. Hastings.
\newblock ``A classical algorithm which also beats
  $\frac{1}{2}+\frac{2}{\pi}\frac{1}{\sqrt{d}}$ for high girth
  max-cut''~(2021).
\newblock  url:~\url{https://doi.org/10.48550/arXiv.2111.12641}.

\bibitem{Shaydulin2023}
Ruslan Shaydulin, Phillip~C. Lotshaw, Jeffrey Larson, James Ostrowski, and
  Travis~S. Humble.
\newblock ``Parameter transfer for quantum approximate optimization of weighted
  {MaxCut}''.
\newblock \href{https://dx.doi.org/10.1145/3584706}{{ACM} Transactions on
  Quantum Computing {\bf 4}, 1--15}~(2023).

\bibitem{boulebnane2023peptide}
Sami Boulebnane, Xavier Lucas, Agnes Meyder, Stanislaw Adaszewski, and Ashley
  Montanaro.
\newblock ``Peptide conformational sampling using the quantum approximate
  optimization algorithm''.
\newblock npj Quantum Information {\bf 9}, 70~(2023).
\newblock  url:~\url{https://doi.org/10.1038/s41534-023-00733-5}.

\bibitem{brandhofer2022benchmarking}
Sebastian Brandhofer, Daniel Braun, Vanessa Dehn, Gerhard Hellstern, Matthias
  H{\"u}ls, Yanjun Ji, Ilia Polian, Amandeep~Singh Bhatia, and Thomas Wellens.
\newblock ``Benchmarking the performance of portfolio optimization with qaoa''.
\newblock \href{https://dx.doi.org/10.1007/s11128-022-03766-5}{Quantum
  Information Processing {\bf 22}, 25}~(2022).

\bibitem{2110.10685}
Sami Boulebnane and Ashley Montanaro.
\newblock ``Predicting parameters for the quantum approximate optimization
  algorithm for max-cut from the infinite-size limit''~(2021).
\newblock  url:~\url{https://doi.org/10.48550/arXiv.2110.10685}.

\bibitem{farhi2022quantum}
Edward Farhi, Jeffrey Goldstone, Sam Gutmann, and Leo Zhou.
\newblock ``The quantum approximate optimization algorithm and the
  {S}herrington-{K}irkpatrick model at infinite size''.
\newblock \href{https://dx.doi.org/10.22331/q-2022-07-07-759}{Quantum {\bf 6},
  759}~(2022).

\bibitem{Dembo2017}
Amir Dembo, Andrea Montanari, and Subhabrata Sen.
\newblock ``Extremal cuts of sparse random graphs''.
\newblock \href{https://dx.doi.org/10.1214/15-aop1084}{The Annals of
  Probability{\bf 45}}~(2017).

\bibitem{crooks2018performance}
Gavin~E Crooks.
\newblock ``Performance of the quantum approximate optimization algorithm on
  the maximum cut problem''~(2018).
\newblock  url:~\url{https://doi.org/10.48550/arXiv.1811.08419}.

\bibitem{streif2020training}
Michael Streif and Martin Leib.
\newblock ``Training the quantum approximate optimization algorithm without
  access to a quantum processing unit''.
\newblock \href{https://dx.doi.org/10.1088/2058-9565/ab8c2b}{Quantum Science
  and Technology {\bf 5}, 034008}~(2020).

\bibitem{zhou2020quantum}
Leo Zhou, Sheng-Tao Wang, Soonwon Choi, Hannes Pichler, and Mikhail~D. Lukin.
\newblock ``Quantum approximate optimization algorithm: Performance, mechanism,
  and implementation on near-term devices''.
\newblock \href{https://dx.doi.org/10.1103/PhysRevX.10.021067}{Physical Review
  X {\bf 10}, 021067}~(2020).

\bibitem{shaydulin2019multistart}
Ruslan Shaydulin, Ilya Safro, and Jeffrey Larson.
\newblock ``Multistart methods for quantum approximate optimization''.
\newblock In IEEE High Performance Extreme Computing Conference.
\newblock \href{https://dx.doi.org/10.1109/hpec.2019.8916288}{Pages 1--8}.
\newblock ~(2019).

\bibitem{Lee2021}
Xinwei Lee, Yoshiyuki Saito, Dongsheng Cai, and Nobuyoshi Asai.
\newblock ``Parameters fixing strategy for quantum approximate optimization
  algorithm''.
\newblock \href{https://dx.doi.org/10.1109/qce52317.2021.00016}{2021 {IEEE}
  International Conference on Quantum Computing and Engineering
  ({QCE})}~(2021).

\bibitem{sack2021quantum}
Stefan~H. Sack and Maksym Serbyn.
\newblock ``Quantum annealing initialization of the quantum approximate
  optimization algorithm''.
\newblock \href{https://dx.doi.org/10.22331/q-2021-07-01-491}{{Quantum} {\bf
  5}, 491}~(2021).

\bibitem{amosy2022iterative}
Ohad Amosy, Tamuz Danzig, Ely Porat, Gal Chechik, and Adi Makmal.
\newblock ``Iterative-free quantum approximate optimization algorithm using
  neural networks''~(2022).
\newblock  url:~\url{https://doi.org/10.48550/arXiv.2208.09888}.

\bibitem{Lykov2020tensorqaoa}
Danylo Lykov, Roman Schutski, Alexey Galda, Valeri Vinokur, and Yuri Alexeev.
\newblock ``Tensor network quantum simulator with step-dependent
  parallelization''.
\newblock In 2022 IEEE International Conference on Quantum Computing and
  Engineering (QCE).
\newblock \href{https://dx.doi.org/10.1109/QCE53715.2022.00081}{Pages
  582--593}.
\newblock ~(2022).

\bibitem{Medvidovic2021QAOA54qubit}
Matija Medvidovi\'c and Giuseppe Carleo.
\newblock ``Classical variational simulation of the quantum approximate
  optimization algorithm''.
\newblock \href{https://dx.doi.org/10.1038/s41534-021-00440-z}{npj Quantum
  Information{\bf 7}}~(2021).

\bibitem{Shaydulin2021Exploiting}
Ruslan Shaydulin and Stefan~M. Wild.
\newblock ``Exploiting symmetry reduces the cost of training {QAOA}''.
\newblock \href{https://dx.doi.org/10.1109/tqe.2021.3066275}{{IEEE}
  Transactions on Quantum Engineering {\bf 2}, 1--9}~(2021).

\bibitem{Shaydulin2020CaseStudy}
Ruslan Shaydulin and Yuri Alexeev.
\newblock ``Evaluating quantum approximate optimization algorithm: A case
  study''.
\newblock \href{https://dx.doi.org/10.1109/IGSC48788.2019.8957201}{Tenth
  International Green and Sustainable Computing Conference}~(2019).

\bibitem{brandao2018concentration}
Fernando G. S.~L. {Brand\~ao}, Michael Broughton, Edward Farhi, Sam Gutmann,
  and Hartmut Neven.
\newblock ``For fixed control parameters the quantum approximate optimization
  algorithm's objective function value concentrates for typical
  instances''~(2018).
\newblock  url:~\url{https://doi.org/10.48550/arXiv.1812.04170}.

\bibitem{Akshay2021}
V.~Akshay, D.~Rabinovich, E.~Campos, and J.~Biamonte.
\newblock ``Parameter concentrations in quantum approximate optimization''.
\newblock \href{https://dx.doi.org/10.1103/physreva.104.l010401}{Physical
  Review A{\bf 104}}~(2021).

\bibitem{Lotshaw2021BFGS}
Phillip~C. Lotshaw, Travis~S. Humble, Rebekah Herrman, James Ostrowski, and
  George Siopsis.
\newblock ``Empirical performance bounds for quantum approximate
  optimization''.
\newblock \href{https://dx.doi.org/10.1007/s11128-021-03342-3}{Quantum
  Information Processing {\bf 20}, 403}~(2021).

\bibitem{Galda2021transfer}
Alexey Galda, Xiaoyuan Liu, Danylo Lykov, Yuri Alexeev, and Ilya Safro.
\newblock ``Transferability of optimal qaoa parameters between random graphs''.
\newblock In 2021 IEEE International Conference on Quantum Computing and
  Engineering (QCE).
\newblock \href{https://dx.doi.org/10.1109/QCE52317.2021.00034}{Pages
  171--180}.
\newblock ~(2021).

\bibitem{lee2022depth}
Xinwei Lee, Ningyi Xie, Dongsheng Cai, Yoshiyuki Saito, and Nobuyoshi Asai.
\newblock ``A depth-progressive initialization strategy for quantum approximate
  optimization algorithm''.
\newblock \href{https://dx.doi.org/10.3390/math11092176}{Mathematics {\bf 11},
  2176}~(2023).

\bibitem{khairy2019learning}
Sami Khairy, Ruslan Shaydulin, Lukasz Cincio, Yuri Alexeev, and Prasanna
  Balaprakash.
\newblock ``Learning to optimize variational quantum circuits to solve
  combinatorial problems''.
\newblock \href{https://dx.doi.org/10.1609/aaai.v34i03.5616}{Proceedings of the
  {AAAI} Conference on Artificial Intelligence {\bf 34}, 2367--2375}~(2020).

\bibitem{1907.05415}
Guillaume Verdon, Michael Broughton, Jarrod~R. McClean, Kevin~J. Sung, Ryan
  Babbush, Zhang Jiang, Hartmut Neven, and Masoud Mohseni.
\newblock ``Learning to learn with quantum neural networks via classical neural
  networks''~(2019).
\newblock  url:~\url{https://doi.org/10.48550/arXiv.1907.05415}.

\bibitem{1911.04574}
Sami Khairy, Ruslan Shaydulin, Lukasz Cincio, Yuri Alexeev, and Prasanna
  Balaprakash.
\newblock ``Reinforcement-learning-based variational quantum circuits
  optimization for combinatorial problems''~(2019).
\newblock  url:~\url{https://doi.org/10.48550/arXiv.1911.04574}.

\bibitem{Wauters2020}
Matteo~M. Wauters, Emanuele Panizon, Glen~B. Mbeng, and Giuseppe~E. Santoro.
\newblock ``Reinforcement-learning-assisted quantum optimization''.
\newblock \href{https://dx.doi.org/10.1103/physrevresearch.2.033446}{Physical
  Review Research{\bf 2}}~(2020).

\bibitem{Alam2020}
Mahabubul Alam, Abdullah Ash-Saki, and Swaroop Ghosh.
\newblock ``Accelerating quantum approximate optimization algorithm using
  machine learning''.
\newblock \href{https://dx.doi.org/10.23919/date48585.2020.9116348}{2020
  Design, Automation {\&} Test in Europe Conference {\&} Exhibition
  ({DATE})}~(2020).

\bibitem{Yao2021}
Jiahao Yao, Lin Lin, and Marin Bukov.
\newblock ``Reinforcement learning for many-body ground-state preparation
  inspired by counterdiabatic driving''.
\newblock \href{https://dx.doi.org/10.1103/physrevx.11.031070}{Physical Review
  X{\bf 11}}~(2021).

\bibitem{Wang2018}
Zhihui Wang, Stuart Hadfield, Zhang Jiang, and Eleanor~G. Rieffel.
\newblock ``Quantum approximate optimization algorithm for {MaxCut}: A
  fermionic view''.
\newblock \href{https://dx.doi.org/10.1103/physreva.97.022304}{Physical Review
  A{\bf 97}}~(2018).

\bibitem{wurtz2021fixedangle}
Jonathan Wurtz and Danylo Lykov.
\newblock ``The fixed angle conjecture for {QAOA} on regular {MaxCut}
  graphs''~(2021).
\newblock  url:~\url{https://doi.org/10.48550/arXiv.2107.00677}.

\bibitem{hadfield2018quantum}
Stuart Hadfield.
\newblock ``Quantum algorithms for scientific computing and approximate
  optimization''~(2018).
\newblock  url:~\url{https://doi.org/10.48550/1805.03265}.

\bibitem{glasserman2004monte}
Paul Glasserman.
\newblock ``Monte carlo methods in financial engineering''.
\newblock \href{https://dx.doi.org/10.1007/978-0-387-21617-1}{Volume~53}.
\newblock Springer. ~(2004).

\bibitem{rudin1974real}
Walter Rudin.
\newblock ``Real and complex analysis''.
\newblock McGraw-Hill. ~(1974).

\bibitem{rudin1976principles}
Walter Rudin.
\newblock ``Principles of mathematical analysis''.
\newblock McGraw-hill. ~(1976).

\bibitem{mcdiarmid_1989}
Colin McDiarmid.
\newblock ``On the method of bounded differences''.
\newblock \href{https://dx.doi.org/10.1017/CBO9781107359949.008}{Page
  148–188}.
\newblock London Mathematical Society Lecture Note Series. Cambridge University
  Press. ~(1989).

\bibitem{warnke_2016}
Lutz Warnke.
\newblock ``On the method of typical bounded differences''.
\newblock \href{https://dx.doi.org/10.1017/S0963548315000103}{Combinatorics,
  Probability and Computing {\bf 25}, 269–299}~(2016).

\bibitem{vershynin_2018}
Roman Vershynin.
\newblock ``High-dimensional probability: An introduction with applications in
  data science''.
\newblock \href{https://dx.doi.org/10.1017/9781108231596}{Cambridge Series in
  Statistical and Probabilistic Mathematics}. Cambridge University Press.
  ~(2018).

\bibitem{Basso_2022}
Joao Basso, David Gamarnik, Song Mei, and Leo Zhou.
\newblock ``Performance and limitations of the {QAOA} at constant levels on
  large sparse hypergraphs and spin glass models''.
\newblock \href{https://dx.doi.org/10.1109/focs54457.2022.00039}{2022 {IEEE}
  63rd Annual Symposium on Foundations of Computer Science ({FOCS})}~(2022).

\bibitem{parisi1980sequence}
G~Parisi.
\newblock ``A sequence of approximated solutions to the s-k model for spin
  glasses''.
\newblock \href{https://dx.doi.org/10.1088/0305-4470/13/4/009}{Journal of
  Physics A: Mathematical and General {\bf 13}, L115}~(1980).

\bibitem{talagrand2006parisi}
Michel Talagrand.
\newblock ``The {P}arisi formula''.
\newblock \href{https://dx.doi.org/10.4007/annals.2006.163.221}{Annals of
  Mathematics}~(2006).

\bibitem{panchenko2014introduction}
Dmitry Panchenko.
\newblock ``The {S}herrington-{K}irkpatrick model''.
\newblock \href{https://dx.doi.org/10.1007/978-1-4614-6289-7}{Springer Science
  \& Business Media}. ~(2013).

\bibitem{shaydulin2021qaoakit}
Ruslan Shaydulin, Kunal Marwaha, Jonathan Wurtz, and Phillip~C Lotshaw.
\newblock ``{QAOAKit}: A toolkit for reproducible study, application, and
  verification of {QAOA}''.
\newblock \href{https://dx.doi.org/10.1109/QCS54837.2021.00011}{Second
  International Workshop on Quantum Computing Software}~(2021).

\bibitem{qaoa_high_girth_arxiv_ver}
Joao Basso, Edward Farhi, Kunal Marwaha, Benjamin Villalonga, and Leo Zhou.
\newblock ``The quantum approximate optimization algorithm at high depth for
  maxcut on large-girth regular graphs and the sherrington-kirkpatrick
  model''~(2021).
\newblock  url:~\url{https://doi.org/10.48550/arXiv.2110.14206}.

\bibitem{herman2022portfolio}
Dylan Herman, Ruslan Shaydulin, Yue Sun, Shouvanik Chakrabarti, Shaohan Hu,
  Pierre Minssen, Arthur Rattew, Romina Yalovetzky, and Marco Pistoia.
\newblock ``Constrained optimization via quantum zeno dynamics''.
\newblock \href{https://dx.doi.org/10.1038/s42005-023-01331-9}{Communications
  Physics {\bf 6}, 219}~(2023).

\bibitem{slate2021quantum}
N.~Slate, E.~Matwiejew, S.~Marsh, and J.~B. Wang.
\newblock ``Quantum walk-based portfolio optimisation''.
\newblock \href{https://dx.doi.org/10.22331/q-2021-07-28-513}{{Quantum} {\bf
  5}, 513}~(2021).

\bibitem{hodson2019portfolio}
Mark Hodson, Brendan Ruck, Hugh Ong, David Garvin, and Stefan Dulman.
\newblock ``Portfolio rebalancing experiments using the quantum alternating
  operator ansatz''~(2019).
\newblock  url:~\url{https://doi.org/10.48550/arXiv.1911.05296}.

\bibitem{Hao2022}
Tianyi Hao, Ruslan Shaydulin, Marco Pistoia, and Jeffrey Larson.
\newblock ``Exploiting in-constraint energy in constrained variational quantum
  optimization''.
\newblock \href{https://dx.doi.org/10.1109/qcs56647.2022.00017}{2022
  {IEEE}/{ACM} Third International Workshop on Quantum Computing Software
  ({QCS})}~(2022).

\bibitem{he2023alignment}
Zichang He, Ruslan Shaydulin, Shouvanik Chakrabarti, Dylan Herman, Changhao Li,
  Yue Sun, and Marco Pistoia.
\newblock ``Alignment between initial state and mixer improves qaoa performance
  for constrained optimization''.
\newblock \href{https://dx.doi.org/10.1038/s41534-023-00787-5}{npj Quantum
  Information {\bf 9}, 121}~(2023).

\bibitem{qiskitfinance}
``Qiskit finance''.
\newblock \url{https://qiskit.org/documentation/finance/}.

\bibitem{NLOpt}
Steven~G. Johnson.
\newblock ``The {NLopt} nonlinear-optimization package''~(2022).
\newblock \url{http://github.com/stevengj/nlopt}.

\bibitem{powell2009bobyqa}
Michael~JD Powell.
\newblock ``The {BOBYQA} algorithm for bound constrained optimization without
  derivatives''.
\newblock Cambridge NA Report NA2009/06{\bf 26}~(2009).

\bibitem{Shaydulin2022}
Ruslan Shaydulin and Stefan~M. Wild.
\newblock ``Importance of kernel bandwidth in quantum machine learning''.
\newblock \href{https://dx.doi.org/10.1103/physreva.106.042407}{Physical Review
  A{\bf 106}}~(2022).

\bibitem{2206.06686}
Abdulkadir Canatar, Evan Peters, Cengiz Pehlevan, Stefan~M. Wild, and Ruslan
  Shaydulin.
\newblock ``Bandwidth enables generalization in quantum kernel models''~(2022).
\newblock  url:~\url{https://doi.org/10.48550/arXiv.2206.06686}.

\bibitem{zhang2022escaping}
Kaining Zhang, Liu Liu, Min-Hsiu Hsieh, and Dacheng Tao.
\newblock ``Escaping from the barren plateau via gaussian initializations in
  deep variational quantum circuits''.
\newblock In Advances in Neural Information Processing Systems.
\newblock Volume~35, pages 18612--18627.
\newblock Curran Associates, Inc.~(2022).

\end{thebibliography}

\section*{Disclaimer}
This paper was prepared for informational purposes with contributions from the Global Technology Applied Research center of JPMorgan Chase \& Co. This paper is not a product of the Research Department of JPMorgan Chase \& Co. or its affiliates. Neither JPMorgan Chase \& Co. nor any of its affiliates makes any explicit or implied representation or warranty and none of them accept any liability in connection with this paper, including, without limitation, with respect to the completeness, accuracy, or reliability of the information contained herein and the potential legal, compliance, tax, or accounting effects thereof. This document is not intended as investment research or investment advice, or as a recommendation, offer, or solicitation for the purchase or sale of any security, financial instrument, financial product or service, or to be used in any way for evaluating the merits of participating in any transaction.
\end{document}